\documentclass[12pt,
nofootinbib,
 amsmath,amssymb,
 aps, article
]{revtex4-2}
\usepackage{amsmath}
\usepackage{amsthm}
\usepackage[english]{babel}
\usepackage{float}
\usepackage{scrextend}
\usepackage[bottom]{footmisc}
\usepackage{comment}
\usepackage{graphicx}
\usepackage{dcolumn}
\usepackage{bm}
\usepackage{setspace}
\usepackage{hyperref}
\usepackage{mathrsfs}
\newtheorem{theorem}{Theorem}
\newtheorem{proposition}{Proposition}
\newtheorem{definition}{Definition}
\newtheorem{corollary}{Corollary}
\newtheorem{conjecture}{Conjecture}

\numberwithin{equation}{section}

\makeatletter
\def\l@subsubsection#1#2{}
\makeatother

\begin{document}

\title{The torus one-point block of 2d CFT and null vectors in $\hat{\mathfrak{sl_2}}$} 

\author{Dario Stocco}

\email{dario.stocco@ipht.fr}
\affiliation{École Polytechnique Paris, ETH Zürich}

\author{Advisor: Sylvain Ribault}
\email{sylvain.ribault@ipht.fr}
\affiliation{IPhT, CEA Saclay}

\date{\today}

\begin{abstract}
This thesis is divided into two parts, where in the first part we investigate the computation of Virasoro 1-point blocks on the torus in the framework of Zamolodchikov's recursion relation. It is widely accepted that this recursion relation contains unphysical poles in the central charge $c$. At each order we conjecture how the pole free expressions depend on the internal and external conformal dimensions and central charge, and propose how to compute it numerically. In this thesis, we have calculated the pole free expression up to order 4. 
In the second part we introduce a conformal field theory with an extra symmetry, described by highest weight representations of the affine Lie algebra $\hat{\mathfrak{sl}}_2$. At level 1, we determined a universal $\mathfrak{sl}_2$ basis-independent 'null operator', which generates null vectors in the usual sense. The 'null operators' are generalized objects and can be applied to any state of the horizontal representation, yielding null vectors and are therefore independent from the choice of horizontal representations.
\end{abstract}

\maketitle

\newpage

\tableofcontents

\newpage

\noindent \textbf{Introduction} 

Two dimensional conformal field theories (shortand \textbf{CFT} or \textbf{2d CFT}) are examples of quantum field theories with a very large amount of symmetry. This large amount of symmetry occurs in two dimensions because of the connection to holomorphic functions and drastically restricts the theory. Those restrictions allows, under certain circumstances, to calculate the correlation functions non-pertubartively and without a Lagrangian description. Therefore 2d CFTs are of large interest and widely used in modern theoretical physics. For instance on one side, it is used to describe statistical phenomenas, like the critical Ising model or percolation theory and on the other side it is very important to study string theories.

Let me give an overview how the thesis is structured. In the first section (\ref{section:cft}) we introduce into important concepts and results of the bootstrap approach to 2d CFTs. Section (\ref{section:recursion}) studies the poles in the central charge $c$ of the recursive representation of 1-point blocks on the torus. I conjecture the functional form of the pole-free expression and propose a method to calculate the unknown coefficients. In the last section (\ref{section:affine}) CFTs with an additional symmetry, described by the affine algebra $\hat{\mathfrak{sl}}_2$, are introduced and the corresponding representation theory is dicussed. I focus on the study of degenerate representations, i.e. on null vectors within highest weight representation. Then the null vectors are applied to three point functions to calculate fusion rules between affine degenerate fields and affine primary fields (\ref{fusionsl2}).

\section{Overview: conformal field theories} \label{section:cft}

Let us introduce in this section the main parts of 2d CFTs. We will give an axiomatic based introduction and follow closely the book of Sylvain Ribault \cite{ribaultplane} and the introductory notes of Bert Schellekens \cite{schellekens}.
\\
The underlying space of the CFTs, discussed in this thesis, are the compactified complex plane $\mathbb{C} \cup \{\infty \}$, known as Riemann sphere and the complex torus $\frac{\mathbb{C}}{\mathbb{Z} + \tau \mathbb{Z}}$ with modular parameter $\tau \in \mathbb{C} - \mathbb{R}$. The discussion mainly concentrates on the Riemann sphere unless stated otherwise.

\subsection{Symmetry algebra}

In the euclidean case we introduce the metric $\mathrm{d}s^2 = \mathrm{d}z \mathrm{d}\bar{z}$ on the Riemann sphere. We consider field theories invariant under conformal transformations, which prominently includes scale transformations $z \mapsto \lambda z$, translations $z \mapsto z + a$ and rotations $z \mapsto e^{i\phi}z$. In general conformal transformations contain any angle preserving coordinate transformations i.e. $g \mapsto \Omega(z)g$ when $z \mapsto f(z)$. Locally this is fullfilled by any holomorphic function $(z,\bar{z}) \mapsto (f(z), \bar{f}(z))$, with $\Omega(z) = f'(z)\bar{f}'(z)$. Here $\bar{z}$ is the complex conjugate of $z$. Expanding a conformal transformation $f(z)$ around the identity map $f(z) = z + \epsilon(z)$, yields the infinite dimensional Witt algebra identified by the holomorphic/antiholomorphic generators $l_n := - z^{n+1}\partial_z$/$\bar{l}_n := - \bar{z}^{n+1}\partial_{\bar{z}}$, with commutation relations,
\begin{gather}
    [l_n,l_m] = (n-m)l_{n+m}, \quad [\bar{l}_n,\bar{l}_m] = (n-m)\bar{l}_{n+m}, \quad [l_n,\bar{l}_m] = 0.
\end{gather}

The global, angle preserving transformations are given by the Möbius group $\mathrm{PSL}_2(\mathbb{C}) = \frac{SL_2(\mathbb{C})}{\{id, -id\}}$ and the identification,
\begin{equation}
     \begin{pmatrix}
     a & b \\
     c & d
     \end{pmatrix}
     \mapsto \frac{az+b}{cz+d}.
\end{equation}

A common rescaling of the matrix elements doesn't change the transformation. Hence it is enough to consider only matrices with determinant one and identify $[g] = \{g, -g\}$ i.e. $g \in \mathrm{PSL}_2$. A change of phase doesn't change obersvables in quantum field theories and therefore we consider Lie group representations $R$ with an additional phase, depending on the multiplication of group elements $R(g)R(g') = e^{i\phi(g,g')}R(gg')$. At the level of the Lie algebra this extra phase is accounted by its central extension. In our case the unique central extended Witt algebra is the Virasoro algebra $\mathfrak{V}$. Formally the Virasoro algebra is defined as the complex vector space spanned by $\{L_{n \in \mathbb{Z}}, c\}$ with Lie brackets,
\begin{gather} \label{virasoro}
    [L_n, L_m] = (n-m)L_{n+m}  + \frac{c}{12} m(m^2 - 1), \quad [L_n, c] = 0.
\end{gather}

$c \in \mathbb{C}$ is the central charge of the theory and commutes with all other elements. Vice versa we have a copy of $\mathfrak{V}$ with the same central charge, to account for the anti-holomorphic part with $[\bar{L}_n,L_m] = 0$. The complete symmetry algebra of space-time transformations of CFTs is the direct sum $\mathfrak{V} \oplus \bar{\mathfrak{V}}$.

\subsection{Representations and fields}

So far we equipped the CFT with a space and a symmetry algebra. In order to be able to calculate observables, we equip it further with representations of $\mathfrak{V} \oplus \bar{\mathfrak{V}}$ and fields transforming in this representations. Vaguely spoken, the space of representations is called the spectrum $\mathbb{S}$ of the CFT. For simplicity we only investigate representations of one copy of the Virasoro algebra and tensor them to get representations of the complete symmetry algebra $\mathfrak{V} \oplus \bar{\mathfrak{V}}$.

\textit{Remark:} From now on we do not distinguish between the action $R(g)$ of the representation and the algebra element $g \in \mathfrak{g}$. Similiarly we do not distinguish between the representation $R$ and the vector space $V_R$ it acts on.

\subsubsection{Highest weight representation}

\textit{Axiom:} We demand that $L_0$ is diagonalizable. 

Let's introduce the Verma module with highest weight vector $|\Delta\rangle$,
\begin{equation} \label{virverma}
    \mathcal{V}_{\Delta} := U(\mathfrak{B}^-)|\Delta\rangle,
\end{equation}

where $U(\mathfrak{B}^-)$ is the universal envelope of the Virasoro algebra with only negative indices $L_{n<0}$. We equip the space $U(\mathfrak{B^-})$ with the basis $\mathfrak{L}:= \{L_N := L_{-n_1} \cdot ... \cdot L_{-n_M}\}_{1\leq n_1 \leq ... \leq n_M}$ and call the integer $N := - \sum_{i=1}^M n_i$ the level. Correspondingly we say that the vector $L_{-n_1} \cdot ... \cdot L_{-n_M}|\Delta \rangle$ is at level $N$.
The action of $\mathfrak{V}$ on the Verma module is then given by,
\begin{gather} \label{virverma2}
    L_{n>0} |\Delta \rangle = 0, \quad L_0 |\Delta\rangle = (\Delta \in \mathbb{C}) |\Delta\rangle.
\end{gather}

The highest weight vector $|\Delta\rangle$ of the representation (\ref{virverma2}) is also called primary state and $\Delta$ its conformal dimension. We will see in section (\ref{conformaltransformation}) the role of the conformal dimension in conformal transformations of primary fields.
With respect to $L_0$ the descendant states at level $N>0$, have eigenvalue $L_0 L_N|\Delta\rangle = (\Delta + N)L_N|\Delta\rangle$. This implies that the real-part of the $L_0$-eigenvalues is bounded from below: $Re(\Delta + N) \geq Re(\Delta)$. We interpret the operator $H \sim L_0 + \bar{L_0}$ as the hamiltonian of the theory and therefore the spectrum is stable, because of the lower bound on the energies $\Delta$.

\subsubsection{Degenerate representations} \label{degrep}

In this section we work out irreducible Verma modules and degenerate representations. In section (\ref{virasorofusion}) we will then see how degenerate representations constrain correlation functions. Verma modules (\ref{virverma}) are reducible if they contain null vectors,

\begin{definition}[Null vector]
    A vector $|\chi,N\rangle := \sum_{|M| = N} a_M L_{-M} |\Delta\rangle$, fullfilling the condition $L_{n>0}|\chi,N\rangle = 0$ is called a null vector at level $N>0$.
\end{definition}

A null vector $|\chi,N\rangle$ generates an invariant subspace $U(\mathfrak{B}^-)|\chi,N\rangle \subset \mathcal{V}_{\Delta}$ and therefore the representation is decomposable. There is a very helpful theorem, connecting null vectors and irreducibility,

\begin{theorem}
    The Verma module is irreducible if and only if it contains no null vectors.
\end{theorem}

Conversely we build irreducible representations (shorthand \textit{irreps}) by taking the quotient of the Verma module with all invariant subspaces generated by all null vectors. The quotient space is called degenerate representation.
For any two positive integers $r,s$ forming the product $rs = N$, exist $p_N := \sum_{rs = N} 1$ inequivalent null vectors $|\chi_{r,s}\rangle$ at level $N = rs$,
\begin{gather} \label{virdeg}
    |\chi_{r,s}\rangle = L_{r,s}|\Delta_{r,s}\rangle, \quad L_{r,s} = \sum_{|M| = N} d_M(c,r,s) L_M.
\end{gather}

The coefficients $d_M$ depend on the central charge $c$ and the integers $r,s$. For example in the case of level 2 null vectors we have,
\begin{equation}
    L_{2,1} = L_{-1}^2 + b^2 L_{-2}, \quad L_{1,2} = L_{-1}^2 + b^{-2} L_{-2},
\end{equation}

where $b$ is introduced in (\ref{I.7}).
Moreover the conformal dimension $\Delta_{<r,s>} = \Delta (P_{<r,s>})$ in (\ref{virdeg}) is determined by the null vector condition on $|\chi_{r,s}\rangle$ and given in terms of the central charge $c$ by the parametrization,
\begin{gather} \label{I.7}
    \Delta(P) := \frac{c - 1}{24} - P^2, \quad P_{<r,s>} := \frac{1}{2}(rb + sb^{-1}), \quad c := 13 + 6 b^2 + 6 b^{-2}.
\end{gather}

$P$ is called the momentum and $b$ the coupling constant. A widely used relation is that $\Delta_{<r,-s>} = \Delta_{<r,s>} + rs = \Delta_{<-r,s>}$, where $\Delta_{<r,-s>}$ is the conformal dimension of the null vector $|\chi_{r,s}\rangle$. In terms of the coupling constant $P_{<r,s>}(b)$ and $\Delta_{<r,s>}(b)$ are invariant under the simultaneous exchange of $r \leftrightarrow s$ and $ b \leftrightarrow b^{-1}$ (Let's name the transformation $b \rightarrow b^{-1}$ \textbf{$b$-symmetry}). 
For generic values of $c$, the null vector $|\chi_{r,s}\rangle$ is the only one in the Verma module $\mathcal{V}_{\Delta_{<r,s>}}$ and thus the quotient space,
\begin{equation}
    R_{r,s} := \frac{\mathcal{V}_{\Delta_{<r,s>}}}{\mathcal{V}_{\Delta_{<r,-s>}}},
\end{equation}

forms a \textit{single} degenerate representation. 

\noindent \textbf{A-series minimal models:}
A very prominent example of CFTs are the A-series minimal models, including the critical Ising model. Their spectra are made of \textit{doubly} degenerate representations, who contain two null vectors and therefore satisfy $\Delta_{<r,s>} = \Delta_{<r',s'>}$ for some $(r,s) \neq (r',s')$. This condition on the conformal dimension determines the central charge $c_{p,q} = 1 - 6 \frac{(q-p)^2}{pq}$ for $p,q$ coprime integers. In particulary A-series minimal models demand the extra condition $2 \leq p, q$ otherwise their spectrum is empty. For example the critical Ising model has central charge $c = \frac{1}{2}$ with $p = 4, q = 3$.

\subsubsection{State field correspondence}

We demand that there exists a one to one map $\mathfrak{i}$ of fields $V(z)$ and states in the spectrum $\mathbb{S}$,
\begin{gather}
    \mathfrak{i}: |v\rangle \mapsto V_{|v\rangle}(z).
\end{gather}

$\mathfrak{i}$ extends the action of $\mathfrak{V}$ onto fields by $L^{(z)}_n V_{|v\rangle}(z) := V_{L_n|v\rangle}(z)$.
In string theory the correspondence can be thought of as having the space of states at the origin of the Riemann sphere $z = 0$, corresponding to the far past on the string worldsheet. The state is translated by $V_{|v\rangle}(z) = e^{-izL_{-1}} V_{|v\rangle}(0) e^{izL_{-1}}$ to any point on the sphere.
\newline
A general representation of the symmetry algebra $\mathfrak{V} \oplus \bar{\mathfrak{V}}$ is realized by tensoring two Verma modules (or equivalently degenerate representations),
\begin{gather} \label{fieldrep}
    \begin{aligned}
        \mathfrak{V} \oplus \bar{\mathfrak{V}} &\rightarrow End(\mathcal{V}_{\Delta} \otimes \mathcal{V}_{\bar{\Delta}}) \\
        L_n + \bar{L_m} &\mapsto (L_n +\bar{L_m})V_{\Delta,\bar{\Delta}}(z) := V_{L_n|\Delta\rangle \otimes |\bar{\Delta}\rangle + |\Delta\rangle \otimes \bar{L_m}|\bar{\Delta}\rangle}(z).
    \end{aligned}
\end{gather}

The field $V_{\Delta,\bar{\Delta}}(z)$ is called \textbf{primary field} and corresponds to the highest weight state $|\Delta\rangle \otimes |\bar{\Delta}\rangle \in \mathcal{V}_{\Delta} \otimes \mathcal{V}_{\bar{\Delta}}$. We call fields diagonal if $\Delta = \bar{\Delta}$ and write $V_{\Delta}$. On the other side a CFT is diagonal if the spectrum contains only diagonal fields. Be careful with the expression $\bar{\Delta}$, as it is not the complex conjugate of $\Delta$ but the conformal dimension corresponding to $\bar{\mathfrak{V}}$.

\textit{Axiom:} We introduce the additional axiom, that the generator $L_{-1}$/$ \bar{L}_{-1}$ acts on fields as the derivative in $z$/$\bar{z}$.

There exists an exceptional field, the Virasoro field $T(y)$, which generates the symmetry algebra $\mathfrak{V}$. The Virasoro field is identified as the energy momentum tensor of conformal symmetries. It is defined implicitly, such that its $n+1$'th moment at $z$ is the generator $L^{(z)}_{n}$ (conversely for $\bar{T}$ generating $\bar{L}^{(z)}_{n}$),
\begin{equation}
    L^{(z)}_n = \frac{1}{2\pi i} \oint_{z} dy (y-z)^{n+1}T(y).
\end{equation}

Therefore $T(y)$ can be written in terms of its moments,
\begin{equation} \label{tmoments}
    T(y) = \sum_{i \in \mathbb{Z}} \frac{L_n^{(z)}}{(y-z)^{n+2}}.
\end{equation}

\textit{Axiom:} $T(y)$ is holomorphic and behaves as $O(y^{-4})$ at $y = \infty$.

On the quantum level the fields are operators, where we are interested in the operator product expansion (shorthand \textbf{OPE}) of $T$ with a primary field $V_{\Delta}$ and $T$ with itself,
\begin{gather} \label{TVOPE}
    \begin{aligned}
    T(y)V_{\Delta,\bar{\Delta}}(z) &\underset{y \rightarrow z}{=} \frac{\Delta V_{\Delta,\bar{\Delta}}(z)}{(y-z)^2} + \frac{\partial_z V_{\Delta,\bar{\Delta}}(z) }{(y-z)} + O(1) \\
    T(y)T(z) &\underset{y \rightarrow z}{=} \frac{c/2}{(y-z)^4} + \frac{2T(z)}{(y-z)^2} + \frac{\partial_z T(z)}{(y-z)} + O(y-z).
    \end{aligned}
\end{gather}

In our case the OPEs (\ref{TVOPE}) have finite radius of convergence and is therefore well defined. The $TT$-OPE is equivalent to the Virasoro commutation relation (\ref{virasoro}) and the $TV$-OPE is generic for primary fields.

\subsection{Observables}

Up to this point we have introduced the underlying space, the symmetry algebra, the representations and fields. Now its time to discuss the observables of the CFT.

Beforehand we introduce a generic notation of the complex modulus, which is often used because of the factorizable left and right chiral representation $\mathfrak{V} \oplus \bar{\mathfrak{V}}$,
\begin{equation}
    |f(\Delta,z,L_n)| ^2 := f(\Delta,z, L_n) f(\bar{\Delta},\bar{z},\bar{L}_n).
\end{equation}

\subsubsection{Correlation functions}

We define the observables to be correlation function of N fields (shorthand \textbf{N-point function}). The correlator is a multilinear function of fields, depending on the conformal dimensions and their position. In the bosonic case the correlator transforms under the one dimensional symmetric representation of the symmetric group $S_N$. Diagonal CFTs are bosonic and contain only spin zero fields (see \ref{conformaltransformation}). We write the N-point function as,

\begin{equation}
    \langle \prod_{i = 1}^N V_{\sigma_i,\bar{\sigma}_i}(z_i)\rangle,
\end{equation}

where we denoted by $(\sigma,\bar{\sigma})$ to indicate any field in the representation (\ref{fieldrep}). The N-point functions are restricted by  Ward identities, which are determined by the Virasoro symmetry of conformal field theory. Ward identities can be obtained by a weighted contour integral of the meromorphic correlator $Z(y) := \langle T(y)\prod_i V_{\sigma_i}(z_i)\rangle$ around all poles $z = z_i$. The weight $\epsilon(y)$ is such that it has no poles outside $\{z_i\}_i$,

\begin{equation} \label{wardint}
    \oint_{\infty} \mathrm{d}y \epsilon(y) Z(y) = 0, \quad \epsilon \underset{y \rightarrow \infty}{\leq} O(y^2).
\end{equation}

This integral is evaluated using the $TV$-OPE (\ref{TVOPE}) and the residue theorem of complex analysis.

\textit{Infinitesimal transformations:} A general result of quantum field theory is that an infinitesimal transformation $f(z) = z + \epsilon(z)$ of a generic field $\Psi(z)$ is given by \cite{schellekens}, 

\begin{equation}
    \delta_{\epsilon}\Psi(z) = \frac{1}{2 \pi i} \oint_{z} dy \epsilon(y) T(y) \Psi(z).
\end{equation}

This statement motivates the use of $Z(y)$ and connects it to conformal transformations of N-point functions.

\subsubsection{Global Ward identities} \label{conformaltransformation}

Let's consider $g \in \mathrm{PSL_2(\mathbb{C})}$ a Möbius transformation and the corresponding infinitesimal transformations $\epsilon \in \{1,y,y^2\}$ in equation (\ref{wardint}). We can deduce how a primary field transforms under $g$,
\begin{equation} \label{fieldtransform}
    V_{\Delta,\bar{\Delta}}'(z) := T_g V_{\Delta,\bar{\Delta}}(z) = |(cz+d)^{-2\Delta}|^2V_{\Delta,\bar{\Delta}}(gz),
\end{equation}

with the action of $PSL_2$ on the space $gz := \frac{az+b}{cz+d}$. 
Under global conformal transformation N-point functions stay invariant, as a consequence of exponentiating the Ward identities (\ref{wardint}) with $\epsilon \in \{1,y,y^2\}$. Therefore the classical conformal invariance transmits anomaly free to the quantum theory,
\begin{equation} \label{conformal_invariance}
    \langle \prod_{i=1}^N V_{\Delta_i,\bar{\Delta_i}}(z_i) \rangle = \langle \prod_{i=1}^N T_g V_{\Delta_i,\bar{\Delta_i}}(z_i) \rangle.
\end{equation} 

The $\mathrm{PSL}_2$ invariance of N-point functions imply, that within correlation functions three coordinates can be choosen freely. Especially 2- and 3-point functions are determined on the whole sphere, if they are known at some fixed coordinates.

The invariance equation (\ref{conformal_invariance}) is generalized to local conformal i.e. holomorphic transformations $z \rightarrow h(z)$,
\begin{equation}
    \langle \prod_{i=1}^N V_{\Delta_i,\bar{\Delta_i}}(z_i) \rangle = \prod_{i=1}^N h'(z) \langle \prod_{i=1}^N V_{\Delta_i,\bar{\Delta_i}}(h(z_i)) \rangle.
\end{equation}

Outside the domain of biholomorphism of $h(z)$, especially at $h^{'}(z_0) \in \{0,\infty\}$, descendant states $L_{n\leq -2}V_{\Delta}$ appear in the correlation function. But we do not go into more details.

\textit{Spin of a primary:} Under space rotations $f(z) = e^{i\phi}z$ the field transforms as $V_{\Delta,\bar{\Delta}}(z) \rightarrow e^{i \phi (\Delta - \bar{\Delta})}V_{\Delta,\bar{\Delta}}(z)$. Let $S := \Delta - \bar{\Delta}$ be the spin of the primary field. For $S \in \frac{1}{2} + \mathbb{Z}$ the field flips sign under $2 \pi = \phi$ rotation and is thus fermionic. For $S \in \mathbb{Z}$ the field is bosonic. Generally the spin can take any value but for the ongoing discussion we dont specify it or consider diagonal fields i.e. spin zero.

\textit{Fields at the north pole:} Another subtlety is how the primary fields and its descendents behave analytically at the point $z = \infty$. From the transformation properties (\ref{fieldtransform}) we deduce that $V_{\Delta}(\infty) = \mathrm{lim}_{z \rightarrow \infty} |z^{2\Delta}|^2 V_{\Delta}(z)$ (choose $gz = -\frac{1}{z}$ and use equation (\ref{conformal_invariance})). From the axiom identifying $L_{-1} = \partial_z$, we can argue that $L_{-1}V_{\Delta}(\infty) = \mathrm{lim}_{z \rightarrow \infty} z^{2\Delta + 1}\bar{z}^{2\bar{\Delta}} L_{-1}V_{\Delta}(z)$, and generally,
\begin{equation}
    L \bar{L}V_{\Delta}(\infty) := \mathrm{lim}_{z \rightarrow \infty} |z^{2\Delta + |L|}|^2 L \bar{L} V_{\Delta}(z).
\end{equation}

\subsubsection{Local Ward identities}

If we choose $\epsilon = \frac{1}{(z-z_i)^{n-1}}, \quad n>1$ in equation (\ref{wardint}) we get a set of local Ward identities. For simplicity we concentrate only on the case where there is at most one descendant field in the N-point function,
\begin{equation} \label{localward}
    \langle L^{(z_i)}_{-n} V_{\sigma_i}(z_i) \prod_{j \neq i} V_{\Delta_j}(z_j)\rangle = \sum_{j \neq i} (-\frac{\partial_{z_i}}{z_{ji}^{n-1}} + \frac{(n-1)\Delta_j}{z_{ji}^n})\langle V_{\sigma_i}(z_i) \prod_{j \neq i} V_{\Delta_j}(z_j) \rangle,
\end{equation}

where we introduce the shorthand notation $z_{ij} := z_i - z_j$. In our setup it can be shown by induction that any N-point function with descendant fields, can be related to N-point functions with lower level descendant fields only by differential operators in $z$.

The global and local Ward identities of the right chiral symmetry $\bar{\mathfrak{V}}$ take a similiar form.

\subsubsection{Operator product expansion}

In the introduction we discussed how conformal field theories are especially interesting because they possess non-perturbative solutions. Hence the $VV$-OPE of primary fields is introduced to be able to solve the CFT exactly. By axiom the $VV$-OPE takes the form,
\begin{equation} \label{OPEform}
    V_{\sigma_1,\bar{\sigma_1}}(z_1) V_{\sigma_2,\bar{\sigma_2}}(z_2) \underset{z_1 \rightarrow z_2}{=} \sum_{\sigma_3,\bar{\sigma_3} \in \mathbb{S}} C^{12}_{3}(z_1,z_2) V_{\sigma_3,\bar{\sigma_3}}(z_2).
\end{equation}

The sum on the right side sums over all fields in the spectrum $\mathbb{S}$ of the CFT. In the case of Virasoro symmetry the $VV$-OPE has finite radius of convergence.
Applying $\oint_{z_1,z_2} \mathrm{d}z (z-z_2)^{n+1}T(z)$ on both sides of (\ref{OPEform}), yields the OPE Ward identities,
\begin{equation}
    \bigg (L_n^{(z_2)} + \sum_{m=-1}^n \binom{m+1}{n+1} z_{12}^{n-m} L_m^{(z_1)}\bigg )V_{\sigma_1}(z_1)V_{\sigma_2}(z_2) = \sum_{\sigma_3 \in \mathbb{S}} C^{12}_{3}(z_1,z_2) L^{(z_2)}_nV_{\sigma_3}(z_2).
\end{equation}

The solution of the OPE Ward identity reads,
\begin{equation} \label{vvope}
    V_{\Delta_1,\bar{\Delta_1}}(z_1) V_{\Delta_2,\bar{\Delta_2}}(z_2) = \sum_{\Delta_3,\bar{\Delta_3} \in \mathbb{S}} \frac{C_{123}}{B_3} \bigg |z_{12}^{\Delta_3^{12}} \sum_{L\in U(\mathfrak{B^-})} z_{12}^{|L|}f^L_{1,2,3}L \bigg |^2V_{\Delta_3,\bar{\Delta_3}}(z_2),
\end{equation}

where we use the notation $\Delta_I^J = \sum_{i \in I} \Delta_i - \sum_{j \in J}\Delta_j$ and the structure constant $B_2$ and $ C_{123}$ of the 2- and 3-point function (see equation (\ref{2p}) and (\ref{3p})). The coefficients $f^L_{1,2,3}$ are the solutions of the linear equations,
\begin{equation}
    \sum_{|L| = N-n} f^L_{1,2,3}(\Delta_3 + N - n + n\Delta_1 - \Delta_2)LV_{\Delta_3}(z) = \sum_{|L|=N}f^L_{1,2,3}L_nLV_{\Delta_3}(z).
\end{equation}

Up to the factor $\frac{C_{123}}{B_3}$ and the spectrum $\mathbb{S}$, the $VV$-OPE of primaries is known only by symmetry considerations. Therefore N-point functions of primaries can be reduced to (N-1)-point functions by the $VV$-OPE.

\subsection{Conformal blocks}

In general N-point functions can be written as a sum of spectrum dependent structure constants $C_J$ and universal $z$-dependent functions $\mathfrak{F}_J^{(N)}(z)$,
\begin{gather}
    \langle \Pi_i V_i(z_i)\rangle = \sum_J C_J \mathfrak{F}_J^{(N)}(z)\mathfrak{F}_J^{(N)}(\bar{z}).
\end{gather}

The functions $\mathfrak{F}^{(N)}$ are the conformal blocks of N-point functions (shorthand \textbf{blocks}). Conformal blocks are universal objects of CFTs, determined only by symmetry considerations and build the theory, as they determine up to the structure constants, all correlation functions.

\subsubsection{2- and 3-point blocks}

For the 2- and 3-point functions there is only one unknown structure constant and the blocks take a simple form. In both cases the correlator is determined by the global Ward identity (\ref{conformal_invariance}).
In the 2-point case the correlator and 2-point block is given by,
\begin{gather} \label{2p}
    \begin{aligned}
        \langle V_{\Delta_1,\bar{\Delta}_1}(z_1) V_{\Delta_2,\bar{\Delta}_2}(z_2)\rangle &= B_3 |\mathfrak{F}^{(2)}(\Delta_1,\Delta_2|z_1,z_2)|^2 \\
        \mathfrak{F}^{(2)} &= \delta_{\Delta_1,\Delta_2} z_{12}^{-2\Delta_1},
    \end{aligned}
\end{gather}

and the 3-point correlator and block by,
\begin{gather} \label{3p}
    \begin{aligned}
        \langle V_{\Delta_1,\bar{\Delta}_1}(z_1)V_{\Delta_2,\bar{\Delta}_2}(z_2)V_{\Delta_3,\bar{\Delta}_3}(z_3)\rangle &= C_{123} |\mathfrak{F}^{(3)}(\Delta_1,\Delta_2,\Delta_3|z_1,z_2,z_3)|^2 \\
        \mathfrak{F}^{(3)} &= z^{\Delta^{12}_3}_{12}z^{\Delta^{23}_1}_{23}z^{\Delta^{31}_2}_{31}. 
    \end{aligned}
\end{gather}

\subsubsection{4-point blocks}

The four point function,
\begin{equation}
    \langle V_{\Delta_1,\bar{\Delta}_1}(z_1)V_{\Delta_2,\bar{\Delta}_2}(z_2)V_{\Delta_3,\bar{\Delta}_3}(z_3)V_{\Delta_4,\bar{\Delta}_4}(z_4)\rangle
\end{equation}

seems to make more problems. First, the global Ward identities (\ref{conformal_invariance}) only determines three degree of freedom. Therefore we are obliged to use local Ward identities i.e. the $VV$-OPE (\ref{vvope}) to determine the undetermined degree of freedom from global Ward identities. Second, in the 4-point case there are 3 possibilities to apply the $VV$-OPE (field 1 with field 2, 3 or 4). Each possibility corresponds to one of the three channel s, t and u. Those channels are related to each other by a permutation of the conformal dimensions $\{\Delta_i\}_{i = 1,..,4}$ and complex coordinates $\{z_i\}_{i = 1,...,4}$. In the s-channel case, where the primaries $V_1$ and $V_2$ are OPE'd, the result reads,
\begin{gather} \label{4point}
    \mathfrak{F}_{\Delta_s}^{(4,s)}(\Delta_1,\Delta_2,\Delta_3,\Delta_4|z_1,z_2,z_3,z_4) = z_{13}^{-2\Delta_1}z_{23}^{\Delta^{14}_{23}}z_{34}^{\Delta^{12}_{34}} z_{24}^{\Delta^{3}_{124}}\bigg(x^{\Delta^{s}_{12}} \sum_{L \in L}x^{|L|}f_{\Delta_1,\Delta_2,\Delta_s}^Lg_{\Delta_s,\Delta_3,\Delta_4}^L\bigg),
\end{gather}

with coefficients $g^L_{\Delta_s,\Delta_3,\Delta_4} = \frac{\langle L V_{\Delta_s}(z_s) V_{\Delta_3}(z_3) V_{\Delta_4}(z_4)\rangle}{C_{s34}} = D_{z_s}(L) \mathfrak{F}^{(3)}(\Delta_s,\Delta_3,\Delta_4|z_s,z_3,z_4)$, where $D_{z_s}(L)$ is the differential operator in $z_s$ associated to the descendant $L \in \mathfrak{L}$.

\textit{Remark:} The term in the bracket of equation (\ref{4point}) correponds to the limit $(z_i)_{i=1,...,4} \rightarrow (x,0,\infty, 1)$ i.e. $\mathfrak{F}^{(4)}_{\Delta_s}(...|x,0,\infty,1)$.

For the sake of completeness, we show how the u-channel block is related to the s-channel block. In the u-channel we apply the OPE to the fields $V_1$ and $V_3$. The resulting u-channel block differs by a permutation from the s-channel block,
\begin{equation}
    \mathfrak{F}^{(4,u)}_{\Delta_u}(\Delta_1,\Delta_2,\Delta_3,\Delta_4|z_1,z_2,z_3,z_4) = \mathfrak{F}^{(4,s)}_{\Delta_u}(\Delta_1,\Delta_3,\Delta_2,\Delta_4|z_1,z_3,z_2,z_4).
\end{equation}

In the limit $(z_i)_{i=1,...,4} \rightarrow (x,0,\infty,1)$ and using the definition $\mathfrak{F}^{(4)}(\Delta_i|x) := \mathfrak{F}^{(4)}(\Delta_i|x,0,\infty,1)$, we get 
\begin{equation}
    \mathfrak{F}^{(4,u)}_{\Delta_u}(\Delta_{i}|x) = x^{-2\Delta_1}\mathfrak{F}^{(4,s)}_{\Delta_u}(\Delta_{i}|\frac{1}{x}),
\end{equation}

where we used the conformal transformation $x \mapsto \frac{0x + i}{ix + 0}$. A similiar statement can be worked out for the t- and s- channel block and the t- and u- channel block.

\subsubsection{Torus: 1-point block} \label{t1p}

If the underlying space is any Riemann surface, in our example the torus, we need to account for the extra structure to calculate the correlation functions.
A torus can be described by a 2d lattice on the complex plane. We identify the points $z \sim w$ if and only if $z = w + n w_1 + m w_2$, where $n,m$ are integers and $w_1, w_2$ linearly independent complex numbers, spanning the lattice. Global conformal invariance allows us to rescale the lattice vectors to be $w_1 = 1$ and $w_2 = \tau \in \mathbb{H}$ in the upper halfplane. Fields must respect the identification $\phi(z) = \phi(w)$ if $w \sim z$. The torus is then the quotient space $\mathbb{C}/(\mathbb{Z} + \tau \mathbb{Z})$ := $\mathbb{C}/\sim$.

In the path integral formalism it is easier to implement the identification $\phi(z) = \phi(w)$ if $z \sim w$. The path integral formalism needs a Lagrangian description, but the result can be generalized and only needs a well defined spectrum $\mathbb{S}$ to trace over,
\begin{equation} \label{torcorr}
    \langle G[\phi](x_i) \rangle \propto \int_{\phi(z) = \phi(w)} \mathrm{D}f G[\phi](x_i) e^{-S_E(\phi)} = \mathrm{tr}_{\mathbb{S}}(G[\phi](x_i)e^{2\pi i \tau (L_0 - \frac{c}{24})}e^{-2\pi i \bar{\tau} (\bar{L}_0 - \frac{c}{24})}).
\end{equation}

For $G := 1$, (\ref{torcorr}) is known as the torus partition function $q^{\Delta - \frac{c-1}{24}}\eta(q)^{-1}$, where $\eta(q) = q^{\frac{1}{24}}\prod_{i \in \mathbb{N}_{*}}(1-q^i)$ is the Dedekind eta function \cite{francesco}. The partition function is invariant under modular transformations $\mathrm{PSL}(\mathbb{Z})$ of $\tau$, which corresponds to the freedom of choosing a lattice base.

Evaluating the 1-point function $G[\phi_{\Delta,\bar{\Delta}}] = \phi_{\Delta,\bar{\Delta}}$ in equation (\ref{torcorr}), in terms of blocks yields ($q := e^{2 \pi i \tau}$),
\begin{equation} \label{tor1pt}
    \langle \phi_{\Delta_1,\bar{\Delta}_1}(0) \rangle = \sum_{\Delta, \bar{\Delta}} C_{\Delta_1,\bar{\Delta}_1,\Delta,\bar{\Delta}} |\mathfrak{T}^{(1)}_{\Delta}(\Delta_1|q)|^2,
\end{equation}

where we introduced the torus 1-point block $\mathfrak{T}^{(1)}$.

A $\Delta$-recursive presentation of the torus 1-point block was proposed by Fateev and Litvinov \cite{fateev},
\begin{gather} \label{torrec}
    \begin{aligned}
        \mathfrak{T}^{(1)}_{\Delta}(\Delta_1|q) &= \frac{q^{\Delta - \frac{c-1}{24}}}{\eta(q)}H_{\Delta}(\Delta_1|q) \\
        H_{\Delta} &= 1 + \sum_{m,n=1}^{\infty}\frac{q^{mn}R_{m,n}(\Delta_1,c)}{\Delta - \Delta_{m,n}}H_{\Delta_{m,-n}},
    \end{aligned}
\end{gather}

where $R_{m,n}$ is defined in equation (\ref{rmn}). This recursion, seen as a series in $q$, makes the poles $\Delta = \Delta_{m,n}$ manifest. The main disadvantage this recursion relation entails, are the additional singularities in the central charge $c$. They are considered unphysical in the sense, that they are expected to vanish after summation. In section (\ref{section:recursion}) the recursion relation is expressed in an explicit form, to make the $c$-singularities present and a method is proposed to calculate a $c$-singularity free expression at each order in $q$.

\subsection{Fusion algebra and fusion rules} \label{virasorofusion}

An important feature of the $VV$-OPE (\ref{vvope}) is the fusion algebra. It encodes the rules which representations (in our case Verma modules (\ref{virverma}) or degenerate representations (\ref{degrep})) occur in the OPE and is thus similiar to the decomposition of tensor products of representations. It is usually denoted as $R_1 \times R_2 = \oplus_i N^i_{12}R_i, N \in \mathbb{N}$ and shares the same property as the OPE, namely it is bilinear, associative and commutative. In this work we are more interested if a certain representation occurs in the fusion product and thus tend to not specify the value of $N^i_{12}$ if it is non-zero.
The most prominent example is the fusion of degenerate representations $R_{r,s}$ with Verma modules $\mathcal{V}_{\Delta}$. Consider the action $L_{r,s}$ on the degenerate field $L_{r,s}V_{\Delta_{<r,s>}} = 0$ within a 3-point function of diagonal fields, 
\begin{equation}
    \langle (L_{r,s} V_{\Delta_{<r,s>}}) V_{\Delta_2}(z_2) V_{\Delta_3}(z_3) \rangle = 0.
\end{equation}

It turns out that this condition constrains the conformal dimension $\Delta_3$ as follow,
\begin{equation} \label{virfusion}
    \Delta(P_3) = \Delta(P_2+ib+jb^{-1}),
\end{equation}

for $i \in I:= \{-\frac{r-1}{2},...,\frac{r-1}{2}\}$ and $j \in J:= \{-\frac{s-1}{2},...,\frac{s-1}{2}\}$. The condition (\ref{virfusion}) is equivalent to the fusion rule,
\begin{equation}
    R_{r,s} \times \mathcal{V}_{\Delta(P)} = \sum_{i\in I}\sum_{j \in J} \mathcal{V}_{\Delta(P+ib+jb^{-1})}.
\end{equation}

Fusion rules in a CFT with the larger symmerty algebra $\hat{\mathfrak{sl}}_2$ will be discussed and derived in chapter (\ref{fusionsl2}).

\section{1-point block on the torus}\label{section:recursion} 

We have seen in section (\ref{t1p}) a recursive formulation (\ref{torrec}) to calculate the 1-point block of a primary on the torus. As already pointed out there will arrise poles in the central charge $c$. In this section we show explicitly that the $c$-poles, of the first four orders in fact cancel and assume they cancel at each order i.e. they are a mathematical artefact of the recursive formulation. 
To start the discussion we introduce the factors $R_{m,n}$, useful reformulations and definitions. It is of importance to keep in mind the three variables we are going to work with: the central charge $c$ in terms of $b$ (\ref{I.7}), the external conformal dimension $\Delta_1$ and the internal conformal dimension $\Delta$. 

\subsection{Preliminaries}

The $R_{m,n}$-factors in equation (\ref{torrec}) are defined with help of the momentum $P_{<r,s>}$ (recall \ref{I.7}),
\begin{gather} \label{rmn}
    R_{m,n} := \frac{2P_{<0,0>} P_{<m,n>}}{\prod_{r = 1-m}^{m} \prod_{s=1-n}^{n} 2P_{<r,s>}} \prod_{r\underset{2}{=}1-2m}^{2m-1} \prod_{s \underset{2}{=}1-2n}^{2n-1} (P_1 + P_{<r,s>}),
\end{gather}

where the subscript 2 under the equation sign indicates, to increment the steps in the product by two instead of one. 

\textit{Note:} We have seen that $P_{<r,s>}$ under $b$-symmetry exchanges $r \leftrightarrow s$ and therefore $R_{r,s}$ has the same feature $R_{r,s}(b) = R_{s,r}(b^{-1})$.

For convenience we split $R_{m,n} = E_{m,n}F_{m,n}$ and write it explicitly in terms of $\beta := b^2$, $\Delta_1$ and $\Delta$,
\begin{gather} \label{rmnpoles}
    \begin{aligned}
        F_{m,n} & := \prod_{r =^2 1}^{2m-1}\prod_{s=^2 1}^{2n-1}(\Delta_{<r,s>} - \Delta_1)(\Delta_{<r,-s>} - \Delta_1) \\
        E_{m,n} & := \frac{m n}{2 m!^2 n!^2} \beta^{2m(n-1)}\prod_{r=1}^{m-1}  \frac{-n^{-2}}{1-\frac{r^2}{n^2}\beta^2} \prod_{s=1}^{n-1} \frac{s^{-2}}{1 - \frac{m^2}{s^2}\beta^2} \prod_{r=1}^{m-1} \prod_{s=1}^{n-1} (\frac{s^{-2}}{1 - \frac{r^2}{s^2}\beta^2})^2.
    \end{aligned}
\end{gather}

The reason of this reformulation is to be able to expand it easily around $\beta = 0$, using $\frac{1}{1-ax} = \sum_{i \geq 0} (ax)^i$. And it makes the poles of $E_{m,n}$ at $\beta \in \{ \pm \frac{1}{1,...,m},..., \pm \frac{n}{1,...,m-1}\}$ visible. Whereas $F_{mn}$ is regular in $\beta$ outside $0,\infty$.

We have already discussed the $b$-symmetry in connection with degenerate representations. $b$-symmetry is manifest with respect to observables because it leaves the central charge (\ref{I.7}) invariant. Thus we define the $b$-symmetric Laurent-polynomials,
\begin{gather} \label{BS}
    B_0 := 1, \quad B_j := \beta^j + \beta^{-j}, \quad j \in \mathbb{N}_{>0},
\end{gather}

because they will help us to formulate the results in a nice way.
Since the recursion sums up terms in the range $1\leq mn\leq N$, it is convenient to define the number $P_N := \sum_{mn \leq N}1$. It counts how many possibilities $(m,n) \in \mathbb{N}^2_{*}$ exist with product $mn \leq N$. As an example we give the first four cases,
\begin{equation}
    P_1 = 1, \quad P_2 = 3, \quad P_3 = 5, \quad P_4 = 8.
\end{equation}


Another common expression at fixed $(m,n) \in \mathbb{N}_{*}^2$, which is used to calculate a $c$-pole free recursion is,
\begin{equation} \label{hprefactor}
    \frac{1}{\prod_{rs \leq N - mn}(\Delta_{m,-n}-\Delta_{r,s})} =(4\beta)^{P_{N-mn}}\prod_{rs \leq N-mn} \frac{1}{s-n + (r+m)\beta} \prod_{rs \leq N-mn} \frac{1}{s+n + (r-m)\beta},
\end{equation}

where we rewrote it, to be able to expand it easily around $\beta = 0$. In this expansion the smallest power of the $\beta$-series appearing in (\ref{hprefactor}), is $P_{N-mn} - (\big \lfloor\frac{N}{n}\big \rfloor - m)$, where we counted the amount of times the condition $s = n$ within the left product on the right-hand-side in (\ref{hprefactor}) is fulfilled. For example in the case $m=1=n$ and $N=4$, we have the smallest $\beta$-power $P_{4-1} - (\big \lfloor\frac{4}{1}\big \rfloor - 1) = 5 - (4-1) = 2$.

\subsubsection{Recursion in powers of $q$}

We rewrite the recursion (\ref{torrec}) in a power series expansion of $q$,
\begin{equation} \label{order}
    H_{\Delta}(P_1|q) = 1 + \sum_{i = 1}^{\infty} q^i H^{(i)}_{\Delta}(P_1).
\end{equation}

At each order $i$ (\ref{order}) the $H^{(i)}_{\Delta}$ can be reexpressed as:
\begin{gather} \label{horder}
    \begin{aligned}
        H^{(1)}_{\Delta}(P_1) &= \frac{R_{1,1}}{\Delta}
         \\
        H^{(2)}_{\Delta}(P_1) &= H^{(1)}_{\Delta = 1} H^{(1)}_{\Delta} + \frac{R_{2,1}}{\Delta - \Delta_{<2,1>}} + \frac{R_{1,2}}{\Delta - \Delta_{<1,2>}} = \frac{R_{1,1}^2}{\Delta} + \frac{R_{2,1}}{\Delta - \Delta_{<2,1>}} + \frac{R_{1,2}}{\Delta - \Delta_{<1,2>}}
        \\
        H^{(3)}_{\Delta}(P_1) &= H^{(2)}_{\Delta = 1} H^{(1)}_{\Delta} + \frac{R_{2,1}R_{1,1}}{(\Delta - \Delta_{<2,1>})\Delta_{<2,-1>}} + \frac{R_{1,2}R_{1,1}}{(\Delta - \Delta_{<1,2>})\Delta_{<1,-2>}} + 
        \frac{R_{3,1}}{\Delta - \Delta_{<3,1>}} + \frac{R_{1,3}}{\Delta - \Delta_{<1,3>}}
        \\
        H^{(4)}_{\Delta}(P_1) &= \frac{R_{1,1}}{\Delta}H^{(3)}_{\Delta = 1} + \frac{R_{2,1}}{\Delta - \Delta_{<2,1>}}H^{(2)}_{\Delta = \Delta_{<2,-1>}} + \frac{R_{1,2}}{\Delta - \Delta_{<1,2>}}H^{(2)}_{\Delta = \Delta_{<1,-2>}} + \frac{R_{3,1}}{\Delta - \Delta_{<3,1>}}H^{(1)}_{\Delta = \Delta_{<3,-1>}} \\
        &+ \frac{R_{1,3}}{\Delta - \Delta_{<1,3>}}H^{(1)}_{\Delta = \Delta_{<1,-3>}} 
        + \frac{R_{4,1}}{\Delta - \Delta_{<4,1>}} + \frac{R_{2,2}}{\Delta - \Delta_{<2,2>}} + \frac{R_{1,4}}{\Delta - \Delta_{<1,4>}}
        \\
        &\vdots
        \\
        H^{(N)}_{\Delta}(P_1) &= \sum_{mn \leq N} \frac{R_{m,n}}{\Delta - \Delta_{m,n}} H^{(N-mn)}_{\Delta = \Delta_{<m,-n>}}(P_1),
    \end{aligned}
\end{gather}

where we define the zero'th order $H^{(0)}_{\Delta} := 1$.
Using equation (\ref{rmn}) we calculate $R_{1,1} = \frac{\Delta_1 (\Delta_1 - 1)}{2}$ and therefore see that there are no poles in the central charge at order 1.

\subsection{Singularity free at order 2: hands on calculation}

The poles in $c$ start to make trouble at order 2 and we go through the steps to get rid of it. Let's consider the expression of $H^{(2)}_{\Delta}$,
\begin{equation} \label{order2}
    H^{(2)}_{\Delta}(P_1) = \frac{R_{1,1}^2}{\Delta} + \frac{R_{2,1}}{\Delta - \Delta_{<2,1>}} + \frac{R_{1,2}}{\Delta - \Delta_{<1,2>}}.
\end{equation}

Altough $R_{1,1}$ is regular, $R_{2,1}(\beta)$ is not and we calculate using (\ref{rmnpoles}),
\begin{equation} \label{order2pole}
    E_{2,1}(\beta) = \frac{1}{4} \frac{1}{\beta^2 - 1}, \quad E_{1,2}(\beta) = \beta^2 E_{2,1}(\beta),
\end{equation}

and therefore the extra poles $\beta = \pm 1$ arise. To cancel the extra poles, we pair the second and third term in (\ref{order2}) and bring them to the same denominator,
\begin{equation} \label{order2pair}
    \frac{\beta E_{1,2}}{(\Delta - \Delta_{<1,2>})(\Delta - \Delta_{<2,1>})} \bigg(\beta^{-1} F_{1,2}(\Delta - \Delta_{<2,1>}) - \beta F_{2,1}(\Delta - \Delta_{<1,2>})\bigg).
\end{equation}

We observe that the factor in brackets in (\ref{order2pair}) changes sign under $b$-symmetry transformation. Laurent-polynomials with this property can be written as a sum of $\beta^n - \beta^{-n} = (\beta - \beta^{-1}) \sum_{i \underset{2}{=}-(n-1)}^{n-1} \beta^i$ and therefore the poles (\ref{order2pole}) cancel i.e. the poles are unphysical.

A more carful analysis and bringing all terms in (\ref{order2}) down to a common denominator, yields the pole free expression at order 2,
\begin{gather} \label{o2}
    H^{(2)}_{\Delta} = \frac{1}{32} \frac{R_{1,1}}{\Delta (\Delta - \Delta_{<1,2>})(\Delta - \Delta_{<2,1>})} \big(96 \Delta^2 + 4\Delta (-14 \Delta_1 + 2 \Delta_1^2 + 5 + c) + c \Delta_1(\Delta_1 - 1)\big).
\end{gather}

This procedure can be applied hypothetically at any order $N$, where we identify terms with the same poles, bring it down to a common denominator and finally cancel the pole. At higher order there arise more and more pole in each term because of the expression $E_{mn}$ in (\ref{rmnpoles}). Therefore it is getting out of control to cancel pole by pole. Nevertheless it has been done up to order 3 and 4, where the reader find the solutions in the appendix (\ref{appA}).

\subsection{Proposal: Singularity free at each order}

To start the discussion, we state first the following conjecture,

\begin{conjecture}
    At each order  $H^{(N)}_{\Delta}$ the poles in the central charge $c$ cancel in between terms and therefore a singularity free expression exists.
\end{conjecture}

Second we define the functions $K_N(\Delta,\Delta_1,\beta)$ in the sense of the recursion (\ref{torrec}),
\begin{gather}
    \begin{aligned} 
        H^{(N)}_{\Delta} &= \frac{R_{11}}{\prod_{kl \leq N}(\Delta - \Delta_{<k,l>})} K_N \\ \label{zamoanddaro}
        K_N &:= \sum_{mn \leq N} \bigg(\frac{E_{mn}F_{mn}}{R_{11}} \prod_{\underset{k,l \neq m,n}{kl\leq N}}(\Delta - \Delta_{<k,l>}) H^{(N-mn)}_{\Delta_{<m,-n>}}\bigg).
    \end{aligned}
\end{gather}

Investigating the algebraic form of the pole free expression of the first four orders (\ref{o2} and appendix \ref{appA}) leads us to conjecture the form of the pole free expression,

\begin{conjecture}
    $K_N$ is a $b$-symmetric Laurent-polynomial in $\beta$ of degree $Q_N := P_N - N$, with coefficients polynomial in $\Delta, \Delta_1$. Concretely the algebraic form of $K_N$ is given by (recall the definition of $B_s$ (\ref{BS})),
    \begin{gather} \label{grandconjecture}
        \boxed{\begin{aligned}
            K_N(\Delta,\Delta_1,\beta) = & \biggl\{\sum_{i=0}^{P_N - 2} \Delta^{P_N - 1-i } \biggl(\sum_{s=0}^{min(Q_N,i)} \sum_{j = 0}^{min(2N-2,2i-2s)} (-1)^j C^N_{ijs} \Delta_1^{j} B_s\biggr) \biggr\}\\
            &+ \Delta_1(\Delta_1 - 1)\sum_{s = 0}^{Q_N} \sum_{j=0}^{2N-4} (-1)^j C^N_{(P_N - 1)js}\Delta_1^j B_s,
        \end{aligned}}
    \end{gather}

    with positive rational coefficients $C^N_{ijs} \in \mathbb{Q}_+$.
\end{conjecture}

Note that in the first line of (\ref{grandconjecture}) the power of $\Delta$ is at least one. Whereas the second line, not included in the sum over $i$, has no $\Delta$-dependence and the power of $\Delta_1$ is at least one.

\subsubsection{Large $c$ limit}

In the large central charge limit of the torus 1-point block, an upper bound on the degree of the Laurent-polynomial $K_N$ can be calculated. Let's work for now with a hypothetical degree $\tilde{Q}_N$ of the Laurent-polynom $K_N$. The large $c$ limit of the block exists and is represented as \cite{largec},
\begin{equation} \label{climit}
    q^{\frac{c}{24} - \Delta} \mathfrak{T}^{(1)}_{\Delta}(\Delta_1|q) = \mathcal{L}_{\Delta}(\Delta_1|q) + \mathcal{O}(c^{-1}),
\end{equation}

where $\mathcal{L}_{\Delta}$ is known as the light torus block.

We observe that in the $\beta \rightarrow 0$ limit, the central charge behaves as $c = 6\beta^{-1} \rightarrow \infty$. And conclude that in the recursive representation (\ref{torrec}), along with the assumption that $K_N$ is a Laurent-polynom of degree $\tilde{Q}_N$, the limit $\beta \rightarrow 0$ exists and is given by,
\begin{gather} \label{betalimit}
    q^{\frac{c}{24} - \Delta} \mathfrak{T}^{(1)}_{\Delta}(\Delta_1|q) \propto \sum_{N \geq 0} \frac{q^N}{\prod_{rs\leq N}(\Delta - \Delta_{r,s})} K_N(\beta) \underset{\beta \rightarrow 0}{\sim} \sum_{N \geq 0} q^N \mathcal{O}(\beta^{Q_N-\tilde{Q}_N}),
\end{gather}

where the dominant behaviour of $\prod_{rs\leq N} (\Delta - \Delta_{r,s})^{-1}$ is given by $\beta^{Q_N}(1 + \mathcal{O}(\beta))$. Therefore we conclude that the limit (\ref{betalimit}) exists only if $\tilde{Q}_N \leq Q_N$.

\subsubsection{Tracking $\Delta$ in (\ref{zamoanddaro})}

By multiplying out the product $\prod_{\underset{k,l \neq m,n}{kl\leq N}}(\Delta - \Delta_{k,l})$ in (\ref{zamoanddaro}), we inspect the prefactors of the $\Delta^i$-monomials for $i \in \{1,...,P_N - 1\}$ in $K_N$ and compare it to the conjecture (\ref{grandconjecture}). Those prefactors are manifest $b$-symmetry invariant \footnote[1]{The $b$-symmetry invariance can be seen if we pair $(m,n)$- with $(n,m)$-terms in (\ref{zamoanddaro})}. Due to the expected, very restricted form of $K_N$, there will arise many non-trivial constraints.
\newline

\textbf{($\Delta^0$):} The coefficient in front of $\Delta^0$ takes a very simple form,
\begin{equation} \label{I.13}
    K_N(\Delta = 0,\Delta_1) = \Delta^0 \times \biggl\{(-1)^{P_N - 1} \frac{\big(\prod_{\underset{k,l \neq 1,1}{kl\leq N}}\Delta_{<k,l>}\big) \cdot R_{11}}{\prod_{rs \leq N - 1}\big(1 - \Delta_{<r,s>}\big)}K_{N-1}(\Delta = 1,\Delta_1)\biggr\}.
\end{equation}

Thus we can check if the power $Q_N$ is consistent and indeed given by the upper bound. We apply the conjecture on $K_{N-1}$, i.e. it is a Laurent-polynomial of degree $Q_{N-1}$ and expand everything else in (\ref{I.13}) around $\beta = 0$ (see \ref{hprefactor}). In the resulting product of $\beta$-series, we observe that the smallest power of $\beta$ appearing is $-Q_N$ and therefore get the expected power of $Q_N$.

In fact, the factor $2 R_{11} = \Delta_1 (\Delta_1 - 1)$ in (\ref{I.13}) implies the zeroes $K_N(\Delta = 0, \Delta_1 = 0, \beta) = 0 = K_N(\Delta = 0, \Delta_1 = 1 , \beta)$, because $K_N$ is regular in $\Delta_1$. Thats why we conclude that no $\Delta^0 \Delta_1^0 \beta^i$-terms appear as conjectured, but the factor $R_{11} \propto \Delta_1 (\Delta_1 - 1)$ appears in the second line of (\ref{grandconjecture}).
\newline

\textbf{($\Delta^{P_N-1}$):} Another interesting identity arises from the prefactor of $\Delta^{P_N - 1}$. Using the conjecture we infer that,
\begin{equation} \label{justanumber}
    \sum_{mn\leq N} \frac{R_{mn}}{R_{11}} H^{(N-mn)}_{\Delta_{<m,-n>}},
\end{equation}

is just a number, independent of $\beta$, $\Delta$ and $\Delta_1$ and thus equal to $C^N_{000}$. Assuming this number exists, we are able to determine a closed formula for it,
\begin{equation}
    C^N_{000} = \sum_{mn = N} n,
\end{equation}

where we took in (\ref{justanumber}) the limit $\Delta_1 \rightarrow 0$ and evaluated the term proportional to $\beta^0$. In the $\Delta_1$-limit only the $mn = N$ - terms contribute, because $R_{mn} \propto \Delta_1(\Delta_1 - 1)$ and $H^{(N-mn)} \propto \Delta_1(\Delta_1 - 1)$ are simply zero (for $\Delta_1 = 0$ the 1-point function (\ref{tor1pt}) is simply equal to the torus partition function and therefore $H_{\Delta}(\Delta_1 = 0, \beta) = 1$ (\ref{order})). And in the final step we determined the $\beta = 0$ - expansion of $E_{mn}F_{mn}/R_{11} \underset{\Delta_1 = 0}{=} n + \mathcal{O}(\beta)$, using the expressions of $E_{mn}$ and $F_{mn}$ (\ref{rmnpoles}). For example $C^N_{000}$ for the first 6 orders is given by,
\begin{equation}
    C^1_{000} = 1, \quad C^2_{000} = 3, \quad C^3_{000} = 4, \quad C^4_{000} = 7, \quad C^5_{000} = 6, \quad C^6_{000} = 12,
\end{equation}

which coincides with the results of the pole free expressions at order 1, 2, 3 and 4 (\ref{o2}), (appendix \ref{appA}).

Tracking in (\ref{justanumber}) the smallest $\beta$-power in the $\beta = 0$ expansion (which occurs at $(m,n) = (1,2)$ \footnote[2]{The number $\lfloor\frac{N}{n}\rfloor - m$ is maximum for positive integers $(1,1) \neq (m,n) = (1,2)$.}) leads to the constraint,
\begin{equation} \label{cons2}
    0 = \big(\beta^{\lfloor N/2 \rfloor - 2} H^{N-2}_{\Delta_{1,-2}}\big)\big|_{\beta = 0}
\end{equation}

if $2 - \lfloor N/2 \rfloor <0$.
For example at order $N = 6$, we might have naively expected a non-zero contribution at $\beta^{2-\lfloor 6/2 \rfloor} = \beta^{-1}$ but it has been numerically checked that (\ref{cons2}) is indeed fullfilled for $N = 6$.

\subsubsection{Calculation of pole free expression}

We have conjectured the existence of a pole free expression and deduced the algebraic form of the $K_N$'s (\ref{grandconjecture}). Here we propose a way to compute the pole free expression from the definition of $K_N$ (\ref{zamoanddaro}),

\begin{corollary}
    The pole free expression in conjecture (\ref{grandconjecture}) can be calculated by expanding $K_N$ (\ref{zamoanddaro}) around $\beta = 0$,
    \begin{equation} \label{howto2}
        \underset{\beta \rightarrow 0}{\mathrm{lim}} K_N = K_N(\Delta,\Delta_1)_
        {Q_N}\beta^{-Q_N} + K_N(\Delta,\Delta_1)_{1-Q_N} \beta^{-Q_N + 1} + ... + K_N(\Delta,\Delta_1)_{0} \beta^{0} + \mathcal{O}(\beta^{1}),
    \end{equation}
    where $K_N(\Delta,\Delta_1,\beta) = \sum_{j=0}^{Q_N} K_N(\Delta,\Delta_1)_j B_j$ is the pole free expression.
\end{corollary}

Expanding (\ref{howto2}) around $\beta = 0$ seems hard to evaluate and thus we want to give a short description how to do it. Recall the definition of $K_N$,
\begin{equation} \label{recallkn}
    K_N = \sum_{mn \leq N} \bigg\{\frac{E_{mn}F_{mn}}{R_{11}} \prod_{k,l \neq m,n}(\Delta - \Delta_{<k,l>}) \bigg( \frac{R_{11}}{\prod_{kl \leq N-mn}(\Delta_{<m,-n>} - \Delta_{<k,l>})} K_{N-mn}\bigg)^{1 - \delta_{N,mn}}\bigg\},
\end{equation}

where we recursively replaced $H^{(N-mn)}_{\Delta_{<m,-n>}}$ by the expression (\ref{zamoanddaro}). The power $1 - \delta_{N,mn}$ needs to be added because $H^{(N-mn)}_{\Delta} = 1$, if $mn = N$. Equation (\ref{recallkn}) is enough to calculate the $\beta = 0$ expansion. We can simply replace the factors $E_{mn}$ by expression (\ref{rmnpoles}) and the factor $\prod_{rs\leq N-mn}(\Delta_{m,-n} - \Delta_{r,s})^{-1}$ by (\ref{hprefactor}). In this substitution (\ref{recallkn}) is completely written as a sum of products of $\beta$-series \footnote[3]{The series are either finite or infinite, but with a lower bound $\gamma > -\infty$ on the powers of $\beta$, $\sum_{i = \gamma}^{\infty} a_i \beta^i$.}. This is because we observe that any other factor, for example $K_{N-mn}$ or $F_{mn}$ (\ref{rmnpoles}), arising in $K_N$, is already a Laurent-polynom in $\beta$. The last step is to evaluate the product of all series and Laurent-polynomials in $\beta$ and to pick up the prefactors of $\beta^{-Q_N}$ up to $\beta^0$. 

\textit{Remark:} There is no need to treat $\Delta_1$ as a variable, since the recursion (\ref{horder}) does not need informations concering $\Delta_1$ and thus can be fixed. It is nevertheless useful to know the $\Delta_1$-dependence, otherwise the $K_N$'s for fixed $\Delta_1$ have to be calculated again for a different choice of $\Delta_1$.

\section{Affine symmetry} \label{section:affine}

In this section we introduce a CFT with an additional symmetry described by the affine Lie algebra $\hat{\mathfrak{g}}$. For a general Lie algebra $\mathfrak{g}$, we define holomorphic currents $J^a(z)$ implicitly through their OPE,
\begin{equation}
    J^a(y) J^b(z) = \frac{k K^{ab}}{(y-z)^2} + \frac{f^{ab}_c J^c(z)}{(y-z)} + O(1),
\end{equation}

where we introduced the level $k$, the structure constant $f^{ab}_c$ of the Lie algebra $\mathfrak{g}$ and the Killing form $K^{ab} := \frac{1}{2g} f^{ac}_d f^{bd}_c$ ($g$ is the dual coxeter number of the Lie algebra $\mathfrak{g}$, in the case of $\mathfrak{sl}_N$ we have $g = N$). By extracting the modes of the current $J^a(y)$,
\begin{equation}
    J^{a,(z)}_n := \oint_z dy (y-z)^n J^a(y),
\end{equation}

we get the generators $J^a_n$ of the affine Lie algebra $\hat{\mathfrak{g}}$, which obey the commutation relations,
\begin{equation} \label{affinecommutation}
    [J^a_n, J^b_m] = f^{ab}_c J^c_{n+m} + n k K^{ab} \delta_{n+m,0}.
\end{equation}

Here the level $k$ becomes the central element of the algebra \footnote[2]{For indecomposable representations of $\hat{\mathfrak{g}}$, the level $k$ acts as a number.}. Do not confuse the level $k$ with the level of descendence $N$. The level of descendence of an element $J^A_N := (J^{a_1}_{-n_1} \cdot ... \cdot J^{a_M}_{-n_M})_{a_i \in A} \in U(\hat{\mathfrak{g}})$, where $A$ is some discrete set, is defined to be the integer $N := \sum_{i=1}^M n_i$. 

\begin{definition}[Horizontal algebra]
    The subalgebra of $\hat{\mathfrak{g}}$ generated by the set $\{J^a_0\}_a$ (we check that the commutations relations (\ref{affinecommutation}) are invariant for $n=0$), is isomorph to the underlying Lie algebra $\mathfrak{g}$. Let's call the subalgebra in this context 'horizontal algebra'.
\end{definition}

Similiar to the Virasoro case a Virasoro field can be constructed to generate the conformal symmetry of the affine symmetric CFT,
\begin{equation} \label{emhat}
    T(z) := \frac{K_{ab}(J^a J^b)(z)}{2(k+g)},
\end{equation}

where we introduced the normal ordered product,
\begin{equation}
    (AB)(z) = \frac{1}{2\pi i}\oint_z \frac{\mathrm{d}z}{y-z}A(y)B(z).
\end{equation}

Equation (\ref{emhat}) is better known as the Sugawara construction and is consistent only if $k \neq -2$. The modes of $T(z)$ fulfill the Virasoro commutation relation (\ref{virasoro}) with central charge $c = \frac{k \mathrm{dim(\mathfrak{g})}}{k+g}$.

\subsection{Affine primary fields and Ward identities}

An affine primary field $\Phi^{R}(z)$ is defined by the OPE with the current $J^a(z)$,
\begin{equation} \label{affineope}
    J^a(y) \Phi^{R}(z) = \frac{-R(t^a)^T\Phi^{R}(z)}{(y-z)} + O(1),
\end{equation}

where $R$ is an arbritrary representation of the horizontal algebra. For convenience lets call $R$ the horizontal representation\footnote[1]{The reader is advised to read appendix (\ref{appB}) to get a short introduction to $\mathfrak{sl}_2$-irreducible representations, which are prominent examples of horizontal representations.}, where the state $|\Phi^{R}\rangle$ corresponding to $\Phi^{R}(z)$, is a state in the horizontal representation. And the action of $R$ (\ref{affineope}) is better understood if written in a basis of the horizontal representation, $\Phi_i \mapsto -R(t^a)^T \Phi_i = -R(t^a)_{ij} \Phi_j$. Note that we expect the OPE $J^aJ^b\Phi$ to be associative and therefore need the minus sign in (\ref{affineope}).
The primary field $\Phi^{R}(z)$ and the Virasoro field $T(z)$ (\ref{emhat}) satisfies the OPE (\ref{TVOPE}) and is thus a primary in the sense of the Virasoro algebra. Therefore correlators of affine primaries fulfill the Virasoro Ward identities and the results from section \ref{section:cft} can be used.

\subsubsection{Isospin variables}
Affine primary fields $\Phi^R$ transform linearly under the representation $R$ of the horizontal algebra $\mathfrak{g}$. Moreover the primaries can be represented as functions $\Phi^R_x$ depending on the isospin variable $x$. $R$ acts then on the isospin variable $x$ as differential operators $D^R_x(t^a) \Phi^R_x := R(t^a) \Phi^R_x$.
In the $\mathfrak{sl}_2$-case a common choice of basis, the $x$-basis, is given by the differential operators,
\begin{gather} \label{xbasis}
    D^j_x(t^-) = -\partial_x, \quad 
    D^j_x(t^0) = x\partial_x - j \quad 
    D^j_x(t^+) = x^2 \partial_x - 2jx. 
\end{gather}

In the initial definition (\ref{affineope}) the action is defined by its transposed, such that in the isospin formalism the consecutive action acts as $J^a_0 J^b_0 \Phi^{R}_x(z) = D^R_x(t^b)D^R_x(t^a)\Phi^{R}_x(z)$.

In the same way, as in the Virasoro case, we can derive affine Ward identities to constrain the N-point functions.

\subsubsection{Global Ward identity}

The global Ward identities for a given basis $\{t^a\}_a$ of the horizontal algebra read,
\begin{equation} \label{sl2global}
    0 = \sum_{i=1}^N D^{R_i}_{X_i}(t^a)\langle \prod^N_{i=1} \Phi^{R_i}_{X_i}(z_i)\rangle.
\end{equation}

In the $x$-basis and the case of $\mathfrak{sl}_2$, the solution of the global Ward identities (\ref{sl2global}) for the three point function is given by,
\begin{equation} \label{affine3p}
    \langle \Phi^{j_1}_{x_1}(z_1)\Phi^{j_2}_{x_2}(z_2)\Phi^{j_3}_{x_3}(z_3) \rangle \propto x^{j^3_{12}}_{12}x^{j^1_{23}}_{23}x^{j^2_{31}}_{31},
\end{equation}
where we again used the short hand notation $j_I^J = \sum_{i \in I} j_i - \sum_{j \in J}j_j$ and $x_{ij} := x_i - x_j$.

\subsubsection{Local Ward identity}

If all except one field in the N-point functions are affine primaries, a useful local Ward identity reads,
\begin{equation} \label{sl2local}
    \langle J^a_{n < 0} \Phi^{\sigma_i}(z_i)\prod_{j\neq i} \Phi^{R_j}_{X_j}(z_j) \rangle = \sum_{j\neq i} \frac{D^{R_j}_{X_j}(t^a)}{z_{ji}^n}\langle \Phi^{\sigma_i}(z_i)\prod_{j\neq i} \Phi^{R_j}_{X_j}(z_j)\rangle.
\end{equation}

Both, global and local Ward identities, along with the three point function (\ref{affine3p}), will be used in section (\ref{fusionsl2}) to calculate fusion rules between affine primary fields and degenerate fields.

\subsection{$\hat{\mathfrak{sl}}_2$ degenerate representations}

In this section we proceed to discuss degenerate representations and determine null vectors in the case of general $\mathfrak{sl}_2$ horizontal representations. Remember a degenerate representation is constructed by modding out the sub-representations, generated by the null vectors. We work with the variable $t := k + 2$.

As for our purpose ($\mathfrak{sl_2}$) we use the $(0,\pm)$-basis with Lie bracket,
\begin{gather}
    [t^0,t^{\pm}] = \pm t^{\pm}, \quad [t^+, t^-] = 2t^0.
\end{gather}

The corresponding symmetric Killing form is given by $K_{00} = 2, K_{+-} = 1$, with inverse $K^{00} = \frac{1}{2}, K^{+-} = 1$ and is used to raise and lower indices. Specifically in the $\mathfrak{sl}_2$ case, the structure constant fulfills the non-trivial identity $f^{ab}_i f^{ic}_d = 2(K^a_d K^{bc} - K^{ac}K^b_d)$. Furthermore the quadratic Casimir element $C_2 = K_{ab}J^a_0J^b_0$ in the $(0,\pm)$-basis is given by $C_2 = 2((J^0_0)^2 + J^+_0 J^-_0 - J^0_0)$.

\subsubsection{Affine highest weight representation}

Let's take an affine Lie algebra $\hat{\mathfrak{g}}$ and a representation $R(t^a)$ of the horizontal algebra $t^a \in \mathfrak{g} \subset \hat{\mathfrak{g}}$. The highest weight representation $\hat{R}$, equivalent to the OPE (\ref{affineope})\footnote[1]{The equivalence can be worked out by determining the action of $J^a_n$ on the primary field $\Phi^{R}$, which is indeed a highest weight representation.}, is constructed by generalizing the action from $\mathfrak{g}$ to $\hat{\mathfrak{g}}$ on any state $|v\rangle \in R$,
\begin{gather} \label{highestweightrep}
    J^a_{n>0}|v\rangle = 0, \quad J^a_0 |v\rangle = -(t_a)^T |v\rangle.
\end{gather}

The resulting, generalized action is a highest weight representation of $\hat{\mathfrak{sl}}_2$. Compared to the Virasoro case (\ref{virverma}), the horizontal subspace is in general not one dimensional and spanned by a basis of the horizontal vector space. 

Null vectors in highest weight representations (\ref{highestweightrep}) of affine Lie algebras are defined generally,

\begin{definition}[Affine null vector]
    A null vector at level $N$ $|\chi,N\rangle = \sum_{i,A} c_A^i J^A_{N} |v_i\rangle$ in a highest weight representation (\ref{highestweightrep}) fulfills the conditions $J^a_{n>0} |\chi,N\rangle = 0$ for any $a$.
\end{definition}

It is enough to check the null vector condition only for $J^a_{1}$ because of the commutation $[J^a_1,J^b_1] = f^{ab}_c J_2^c$ and by induction on the level.

The connection to irreducibility is given by the following theorem \cite{bauer},

\begin{theorem}
    The highest weight representation (\ref{highestweightrep}) is irreducible if and only if it contains no null vectors.
\end{theorem}

Moreover the existence of null vectors imply an $\mathfrak{sl}_2$-module,

\begin{proposition} \label{nullmodule}
    The left-action of $\mathfrak{sl}_2$ on null vectors yields other null vectors at the same level.
\end{proposition}

\begin{proof}
    Let $|\chi,N\rangle$ be a null vector. We check for the null vector condition, $J^m_1J^a_0|\chi,N\rangle = ([J^m_1,J^a_0] + J^a_0 J^m_1) |\chi,N\rangle$ = 0, where we use $[J^m_1,J^a_0] = f^{ma}_e J^e_1$.
\end{proof}

The goal in the next sections is to find null vectors at level 0 and 1, to be able to reproduce the fusion rules of degenerate representations at any level similiar to section (\ref{virasorofusion}).

\subsubsection{Level 0 null vectors} \label{lvl0}

The highest weight representation extension of the finite irreducible representation with half integer spin $j = j_{0,s} := \frac{s-1}{2}, \quad s \in \mathbb{N}_+$ (see appendix \ref{appB}) of $\mathfrak{sl}_2$ contains a well known null vector at level 0 \cite{lorentz1}, 
\begin{equation}
    (J^-_0)^{2j+1} |j,j\rangle = 0.
\end{equation}

This null vector satisifes trivially the null vector condition and is actually of different nature, since it is equal the $0$-element of the representation. Consequently, we don't need to mod out the subspace, in order to derive fusion rules.

\subsubsection{Level 1 null vectors} \label{lvl1}

Malikov, Feigin and Fuks already determined examples of null vectors at each level $N \geq 1$, but they worked them out in a specific $\mathfrak{sl}_2$ basis and horizontal representation~\cite{feigin}. It is therefore of interest to work out Lie algebra basis independent null vectors in a horizontal representation-independent framework. Hence the question arise, if ideas in this framework around $\mathfrak{sl}_2$ can be extended to the Lie algebra $\mathfrak{sl}_N$ or general simple Lie algebras.
We go to a general setting and focus on horizontal representations, which fulfill two properties,
\newline

\noindent \textbf{(1)} \textit{The horizontal representation is indecomposable and thus the quadratic Casimir $C_2$ is proportional to the identity. In the case of $\mathfrak{sl}_2$ we use the parametrization $C_2 = 2j (j+1)$ with reflection symmetry $j \mapsto -1-j$.} \\
\textbf{(2)} \textit{The horizontal representation $R$/$V_R$ canonically extends to a highest weight representation $\hat{R}$/ $\hat{V}_R$.} \\

This leads to the natural definition,

\begin{definition}[Null vector space] \label{NVA}
Consider $U(\hat{\mathfrak{sl}}_2)$ where we mod out the subvectorspaces $W_1 := \mathrm{span}_{\mathbb{C}}(U(\hat{\mathfrak{sl}}_2)J^a_{n>0})$ and $W_2 := \mathrm{span}_{\mathbb{C}}(U(\hat{\mathfrak{sl}}_2)(C_2 - 2j(j+1)))$. We call the resulting vector space $\frac{U(\hat{\mathfrak{sl}}_2)}{W_1 + W_2}$, the null vector space. 
\end{definition}

Remember in the universal enveloping framework two elements are equal if they differ by commutation. For example in $W_2$ we have $J^+_0 J^0_0 J^-_0 = J^0_0J^+_0 J^-_0 - J^+_0J^-_0$ and thus can be replaced by $J^+_0 J^-_0 = j(j+1) - (J^0_0)^2 + J^0_0$. And in $W_1$ the element $J^+_1 J^-_0 = 2 J^0_1 + J^-_0 J^+_1$ is simply zero.

The null vector space encodes the action of the affine highest weight representation, i.e. for any $T \in \frac{U(\hat{\mathfrak{sl}}_2)}{W_1 + W_2}$ we define the canonical map,
\begin{equation} \label{map2}
    T: V_R \rightarrow \hat{V}_R, \quad |v\rangle \mapsto T|v\rangle.
\end{equation}

We also need to redefine the null vector condition accordingly,

\begin{definition}[Null operator]
    Within the null vector space, the null vector condition transforms into $\forall a: J^a_1 \hat{T} = 0$ or equivalently $\forall a: [J^a_1,\hat{T}] = 0$. $\hat{T}$ is  called a null operator and maps any horizontal state to a null vector via the corresponding map (\ref{map2}).
\end{definition}
 
Left-action of $U(\mathfrak{sl}_2)$ on $\hat{T}$ and the freedom of choice of $|v\rangle \in V_R$ forms a set of null operators $U(\mathfrak{sl}_2) \hat{T} U(\mathfrak{sl}_2)$. This space actually forms a vector space over $\mathbb{C}$ and is an $\mathfrak{sl}_2$-module by means of the adjoint action.
The main advantage of the null vector algebra is that we are able to search for null operators, without taking care about the affine highest weight representation.
\newline 

\noindent \textbf{Eigenspaces:}

The null vector space forms an $\mathfrak{sl}_2$-module and therefore we can diagonalize it with respect to the adjoint action $[J^0_0,\circ]$. The resulting eigenspaces $E^Q$ contains elements of constant charge $Q$, where the charge of an element $\prod_{i = 1}^M J^{a_i}_{n_i} \in \frac{U(\hat{\mathfrak{sl}}_2)}{W_1 + W_2}$ is given by $Q := \sum_{i=1}^M a_i$. Such that for an element $J^Q \in E^Q$ we have $[J^0_0, J^Q] = Q J^Q$.
\newline

\noindent \textbf{Universal null vector:}

At level 1 we have determined a basis independent null operator,
\begin{equation} \label{level1null}
    \boxed{\hat{T}^c_1 := 2 K_{ab} J^a_{-1}J^b_{0} J_0^c - 
     t f^c_{ab} J^a_{-1}J^b_0 - t^2 J^c_{-1},}
\end{equation}

where the quadratic Casimir takes the value $C_2 = \frac{t^2}{2} - t$. We want to show explicitly that $\hat{T}^c_1$ fulfills the null vector condition. To do so let's act with $J^d_1$ and use the commutation relation (\ref{affinecommutation}),
\begin{gather}
    \begin{aligned}
        J^d_1 \hat{T}^c_1 &= 2K_{ab}[J^d_1,J^a_{-1}]J^b_0J^c_0 - t f^c_{ab}[J^d_1,J^a_{-1}] J^b_0 - t^2 [J^d_1,J^c_{-1}] \\
        &= (f^{da}_e J^e_0 + k K^{da})(2K_{ab}J^b_0 J^c_0 - t f^c_{ab}J^b_0) - t^2 (f^{dc}_e J^e_0 + k K^{dc}) = ... 
    \end{aligned}
\end{gather}

Then we use the identities $f^{d}_{be}J^e_0J^b_0 = 2 J^d_0$ and $f^{ab}_i f^{ic}_d = 2(K^a_d K^{bc} - K^{ac}K^b_d)$ (this identity only applies to $\mathfrak{sl}_2$) and group together similiar terms,
\begin{gather} \label{step2}
    \begin{aligned}
        ... = J^d_0J^c_0(4 + 2k) -2tJ^c_0J^d_0 + K^{cd}(2t C_2 - t^2 k) + f^{cd}_e J^e_0(\underbrace{t^2 - tk - 2t}_{= 0}).
    \end{aligned}
\end{gather}

Finally we commute the term $-2t J^c_0J^d_0 = -2t (J^d_0J^c_0 + f^{cd}_e J^e_0)$ and get that (\ref{step2}) is zero if and only if $C_2 = \frac{tk}{2} = \frac{t^2}{2} - t$.
Interestingly the null vector operators $T^c_1$ satisfy the commutation relation,
\begin{equation} \label{nullcommutation}
    [J^a_0,T^b_1] = f^{ab}_c T^c_1.
\end{equation}

\noindent \textbf{Set of null operators:}

The main goal is to give a sketch of the proof of the following conjecture, where we work for simplicity in the $(0,\pm)$-basis,

\begin{conjecture}[Universal null operator] \label{conjuniq}
    The null operator $\hat{T}^0_1$ at level 1 and $C_2 = \frac{t^2}{2} - t$ is universal, such that the subspace $\mathrm{span}_{\mathbb{C}}(U(\mathfrak{sl}_2)\hat{T}^{0}_1U(\mathfrak{sl}_2))$ contains the set of all null operators in the null vector space $\{\hat{T}\in \frac{U(\hat{\mathfrak{sl}}_2)}{W_1 + W_2}| J^a_1 \hat{T} = 0\}$.
\end{conjecture}

To begin with, we want to show the weaker proposition,

\begin{proposition} \label{prope0}
    $\hat{T}^0_1$ generates any finite null operator in $E^0$. The null operators are generated by $\hat{T}^0_1 \mathrm{span}_{\mathbb{C}}(J_0^0)$.
\end{proposition}

In general we expect, that every null operator can be written as a linear combination of products $\prod^N_{i = 1} J^{a_i}_{n_i}$ and therefore the following proposition helps,

\begin{proposition}
    Every product $\prod^N_{i = 1} J^{a_i}_{n_i}$ at level $(-\sum_i n_i) = 1$ can be rewritten as a linear combinations of products $J^{c}_{-1} \prod_j J^{b_j}_0$ such that the charge $c + \sum_j b_j = \sum_i a_i$ stays invariant.
\end{proposition}

\begin{proof}
    Identify all the generators $J^a_{n>0}$ and commute them to the right, to be left with generators at level 0 and one generator at level 1. Commute the level 1 generator to the left. All this operations conserve the level and charge due to the commutation relation (\ref{affinecommutation}).
\end{proof}

In addition we conclude that if $\hat{T} = \lambda_1 \prod_i J^{a_i}_{n_i} + \lambda_2 \prod_i J^{b_i}_{m_i}$ is a null operator with charge $A\sum a_i \neq \sum b_i$, then each term is a null operator by its own, because of linear independence of vectors in different eigenspaces $E^Q$. Therefore at each charge and level, we can rearrange the null vector candidate into a simple ansatz. The ansatz at level one and in the Eigenspace of zero charge $E^0$ can generally be written as,
\begin{equation} \label{an1}
    \hat{T} = J^+_{-1}J^-_0 (\sum_{i\geq 0}a_i (J^0_0)^i) +  J^0_{-1}(\sum_{i\geq 0}b_i (J^0_0)^i) +  J^-_{-1}J^+_0 (\sum_{i\geq 0}c_i (J^0_0)^i).
\end{equation}

Applying the null operator condition $J^a_1 \hat{T} = 0$ to this ansatz yields the solution,
\begin{equation} \label{nullrecursion}
    b_{i+1} = 2 a_i - t a_{i+1}, \quad c_{i+1} = \frac{2}{t}(a_i - c_i)  - a_{i+1},
\end{equation}

where we define that quantities with negative indices vanish.
\newline

\noindent \textbf{Finite solutions:}

A finite solution of (\ref{nullrecursion}) is defined by the condition $a_{N + (i \geq 1)} = 0$ with $N \in \mathbb{N}_{*}$, which implies $b_{N + 1 + (i \geq 1)} = 0$ and $c_{N+(i \geq 1)} = 0$. The recursion (\ref{nullrecursion}) constrains $a_N = \frac{2}{t}(a_{N-1} - c_{N-1}(a_{i\leq N-1}))  - a_{N}$ and therefore we can freely choose $A_N := (a_0,...,a_{N-1})$, to determine any finite solutions for any $N$. We write $\hat{T}_N(A_N)$ for the null vector operator corresponding to the solution $A_N$. 
The space of finite solutions $\{\hat{T}_N(A_N)| A_N \in \mathbb{C}^{N}, N \in \mathbb{N}\}$ forms a vector space because the recursion (\ref{nullrecursion}) is linear.

We will give two examples of finite solutions at $N = 1,2$,
\begin{gather}
    \begin{aligned}
        \hat{T}_1(t,2) =& J^0_{-1}(4(J^0_0)^2 - t^2) + J^-_{-1}J^+_0(2J_0^0 - t) + J_{-1}^+J_0^-(2J_0^0 + t) \\
        \hat{T}_2(t,4,4/t) =& J^0_{-1}(-t^2 - 2t J^0_0 +4(J^0_0)^2 + 8/t (J^0_0)^3) + J^-_{-1}J^+_0(-t + 4/t(J^0_0)^2) \\
        &+ J_{-1}^+J_0^-(t + 4 J^0_0 + 4/t (J^0_0)^2).
    \end{aligned}
\end{gather}

Actually we have the redundancy $\hat{T}_2= \hat{T_1} (2J_0^0 + t)/t$. The solution $\hat{T}_1(t,2)$ corresponds to the initial result (\ref{level1null}) with $c = 0$.
After working out the finite solutions in $E^0$, we are ready to proof proposition \ref{prope0},

\begin{proof}
    The proof is done by induction on the positive integer $N$, such that $a_{N+i} = 0$ for $i\geq 1$. We have already seen that at $N=2$ it can be related to $N=1$ and thus assume it is possible up to $N-1$.
    \\
    We start by noticing that, due to the linear character of the solution $(A_N,a_N)$ we can choose without loss of generality $a_0 = 1$ and write $(A_N,a_N) = (1,0,...,0,a_N') + (0,a_1,...,a_N - a_N')$. Because of $0 = a_0 = c_0 = b_0$ in the second term, we can factor out $a_1 J^0_0$ in the ansatz (III.15) and relate it to the $N-1$ case.
    \\
    To decompose the first term $(1,0,...,0,a_N')$, we notice first that the recursion (\ref{nullrecursion}) implies $a_N' =(-1)^{N+1} (\frac{2}{t})^N$. Therefore we decompose it further $(1,0,...,0,a_N') = (1,0...,0,a_{N-1}') + (-1)^{N+1}(\frac{2}{t})^{N-1} (0,...,0,1,\frac{2}{t})$ and factor out $(J_0^0)^{N-1}$ in the second term. We see both terms can be related to lower $N$-cases. This construction is equivalent to $\hat{T}(1,0,...,a_N) = \hat{T}(1,0,...,a_{N-1}) + (-1)^{N+1}(\frac{2}{t})^{N-1} \hat{T}(1,\frac{2}{t})(J_0^0)^{N-1}$. Using the induction assumption again, we conclude that we can decompose every finite solution to $\hat{T}^0_1 \sum_i a_i (J_0^0)^i$.
\end{proof}

In different eigenspaces with non-zero charge $E^Q$ we can create null operators by acting multiple time with $J^{\pm}_0$ on $\hat{T}^0_1$ from the left or right. The last step would be to proof that at each charge $ Q \neq 0$ all finite solutions can be generated by $\hat{T}^0_1$. We haven't found a promising proof, but it should closely follow the proof of proposition \ref{prope0} and therefore leave it as a conjecture,

\begin{conjecture}
    The null vectors in each eigenspace $E^Q$ with charge $0 \neq Q \in \mathbb{Z}$ are generated by the unique element $T^{0}_1$.
\end{conjecture}

\textit{Remark:} Note that indecomposable horizontal representations do not contain a unique vector to generate the whole vector space by repeatedly applying the action, whereas an irreducible does. Therefore we expect that every null vector in the irreducible case, can be reached with the map (\ref{map2}) and the subspace $\mathrm{span}_{\mathbb{C}}(U(\mathfrak{sl}_2)\hat{T}^{0}_1U(\mathfrak{sl}_2))$ from conjecture \ref{conjuniq}. It is not entirely clear if the same holds for indecomposable representations.

\subsection{$\hat{\mathfrak{sl}}_2$ fusion rules} \label{fusionsl2}

As in the Virasoro case, we deduce the analog of fusion but with primary fields in $\hat{\mathfrak{sl}}_2$-degenerate representations (shorthand \textit{degenerate field}). Generally we name the null operator $\hat{T}^{r,s}$ and the degenerate representation $\hat{R}^{r,s}$, if they correspond to the spin $j_{r,s} = \frac{s-1}{2} - \frac{t}{2} r$. On the other side we name the highest weight representation $R_j$ with well defined qadratic Casimir $2j(j+1)$ and spin $j$.
The idea is then to use the prominent relation $0 = \hat{T}^{r,s}|v\rangle$ of degenerate representations within N-point functions,
\begin{equation} \label{nullvectorequation}
    0 = \langle \hat{T}^{r,s} \Phi_{j_{r,s}} \Phi^{R_1}_{j_1}\cdot ... \cdot \Phi^{R_{N-1}}_{j_{N-1}} \rangle,
\end{equation}

where the degenerate field $\Phi_{r,s}$ corresponds to the state $|v\rangle$. Then we apply the affine Ward identities (\ref{sl2global}) resp. (\ref{sl2local}) to deduce the differential equation of (\ref{nullvectorequation}) in the isospin framework. We call these equations (\ref{nullvectorequation}) \textit{null vector equations}. Take care that the consecutive action acts as $J^a_0 J^b_0 = (-D(t^b))(-D(t^a))$.

\subsubsection{Finite level 0}

For the finite irreducible horizontal representation (\ref{lvl0}), we get the null vector equation $(J^-_0)^{2j + 1}|j,j\rangle = 0$ with spin $j = j_{0,s}$. Using the global Ward identity (\ref{sl2global}) and the $x$-basis (\ref{xbasis}), this constrains the 3-point function (\ref{affine3p}),
\begin{equation}
    0 = (\partial_x)^{2j+1} \langle \Phi^j_x(z) \Phi^{j_2}_{x_2}(z_2) \Phi^{j_3}_{x_3}(z_3) \rangle,
\end{equation}

For a non-zero 3-point function the condition on the spins read,
\begin{equation}
    \prod_{i=0}^{2j}(j_3 - j_2 - j_{0,s} + i) = 0.
\end{equation}

This can be identified as the fusion rule of a generic field with a degenerate field with spin $j_{0,s}$,
\begin{equation}
    \hat{R}^{0,s} \times \hat{R}_j = \hat{R}_{j + j_{0,s}} + \hat{R}_{j + j_{0,s-2}} + ... + \hat{R}_{j - j_{0,s}}.
\end{equation}

Thus in general the field with spin $j_{0,1} = 0 \Rightarrow \Delta_{j = 0} = 0$ \footnote[1]{The conformal dimension of an $\mathfrak{sl}_2$ primary is given by $\Delta_j = \frac{j(j+1)}{t}$} acts as the identity field $\Phi^0(z) \propto id$.

\subsubsection{Generic representation level 1}

Altough Malikov, Feigin and Fuks have determined a null vector at level 1, we would not be able to use it directly to determine the null vector equation. This is because it is written in a basis which diagonalizes $J^0_0$. Therefore we take the degenerate representation with null operator $\hat{T}_1^0$ (\ref{level1null}) on any horizontal state $\hat{T}^0_1|v\rangle$ and spin $j_{1,1}$. This approach yields the following null vector equation in the case of the 3-point function,
\begin{gather}
    \begin{aligned}
        0 = \sum_{s = 2,3} \frac{1}{z_{s1}} \bigg(2K_{ab}D_{x_1}(t^0)D_{x_s}(t^a)D_{x_1}(t^b) + t f^0_{ab} D_{x_s}(t^a) D_{x_1}(t^b)& \\
        - t^2 D_{x_s}(t^0)\bigg) \langle \Phi^{j_{1,1}}_{x_1}(z_1) \Phi^{j_2}_{x_2}(z_2)\Phi^{j_3}_{x_3}(z_3)\rangle&.
    \end{aligned}
\end{gather}

In the limit $z_1, z_3 \rightarrow 0, \infty$ only the term proportional to $\frac{1}{z_{21}}$ stays alive and in the $x$-basis it becomes the differential equation,
\begin{gather}
    \begin{aligned}
        0 = \big(&-2x_1 x_{12}^2 \partial_2(\partial_1)^2 + 4j_2 x_1 x_{21} (\partial_1)^2 + 4x_1(1+t)x_{21} \partial_2\partial_1\\
        &- 4j_2 x_1 (1+t)\partial_1 - 2x_1 t(1+t) \partial_2\big) \langle \Phi^{j_{1,1}}_{x_1}(z_1) \Phi^{j_2}_{x_2}(z_2)\Phi^{j_3}_{x_3}(z_3) \rangle.
    \end{aligned}
\end{gather}

We use the solution of the 3-point function (\ref{affine3p}) and get the following condition on the spins,
\begin{gather}
    \begin{aligned}
        0 = \big(&8 j_2^3 + 8 j_3^3 - 4 j_3^2 (t-2) - 2 j_3 t^2 + (t-2)t^2 \\
        &- 4 j_2^2 (2 j_3 + t - 2) - 2 j_2 (4 j_3^2 - 4 j_3 (t - 2) + t^2)\big).
    \end{aligned}
\end{gather}

If the 3-point function is non-zero, this equation has the three solutions $j_3 \in \{j_2 \pm j_{1,1}, -1-j_2 + j_{1,1}\}$ \footnote[2]{The parametrization $C_2 = 2j(j+1)$ contains the symmetry $j \rightarrow -1 - j$ and we get equivalent solutions up to this symmetry.}. Therefore the fusion rules can be determined,
\begin{equation}
    \hat{R}^{1,1} \times \hat{R}_j = \hat{R}_{j+j_{1,1}} + \hat{R}_{j-j_{1,1}}.
\end{equation}

\subsubsection{Generalized Fusion Rule}

The Fusion algebra is associative and commutative, such that we are able to generate the fusion rule of a degenerate field with spin $j_{r,s}$ by repeatedly fusing fields with spin $j_{0,2}$ and $j_{1,1}$  (the field corresponding to $j_{0,1} = 0$ is the identity field $\hat{R}^{0,1} \times \hat{R}_j = \hat{R}_j$).
To see that these two spins are enough to generate the fusion rules, we take a look at the fusions,
\begin{equation}
    \hat{R}^{1,1} \times \hat{R}^{1,1} = \hat{R}^{2,1} \oplus \hat{R}^{0,1}, \quad \hat{R}^{0,2} \times \hat{R}^{0,2} = \hat{R}^{0,3} \oplus \hat{R}^{0,1}, \quad  \hat{R}^{0,2} \times \hat{R}^{1,1}  = \hat{R}^{1,2}.
\end{equation}

Thus $\hat{R}^{1,1}$ increases or decreases the left index by one unit and $\hat{R}^{0,2}$ increases the right index by one unit. Using this we can write down, how a generic highest weight representation $\hat{R}_j$ fusions with a degenerate representation $\hat{R}^{r,s}$:
\begin{equation} \label{sl2fusion}
    \hat{R}^{r,s} \times \hat{R}_j = \hat{R}_{j - j_{r,s}} + \hat{R}_{j - j_{r-2,s}} + \hat{R}_{j - j_{r,s-2}} + ... + \hat{R}_{j - j_{-r,-s + 2}}
\end{equation}

The fusion between two degenerate fields, respecting associativity and commutativity, yields the closed formula,
\begin{equation}
    \hat{R}^{r,s} \times \hat{R}^{a,b} = \sum_{i \underset{2}{=} |r-a|}^{r+a} \sum_{j \underset{2}{=} |s-b|+1}^{s + b - 1} \hat{R}^{i,j}.
\end{equation}

The derived fusion rule (\ref{sl2fusion}) coincide with previous work, for example in \cite{ribaultplane}, \cite{bauer}, \cite{awata}.

\section*{Conclusion}

In this thesis we have investigated the poles in the central charge the torus 1-point block and null vectors in highest weight representation of $\hat{\mathfrak{sl}}_2$.

By calculating the first four orders of the recursion relation of the 1-point block (\ref{horder}), we have found that there seems to be a very specific algebraic form of the $c$-pole free expression $K_N$ (\ref{zamoanddaro}). It is a laurent polynomial in the variable $\beta := b^2$ of degree $Q_N := \big(\sum_{1 \leq mn \leq N} 1 \big)- N$, with prefactors depending on the external $\Delta_1$ and internal $\Delta$ conformal dimension. Altough it seems hard to get rid of the poles in the central charge, we have found that an expansion around $\beta = 0$ of $K_N$, can be used to find the unknown prefactors of the Laurent-polynomial. The possibility to calculate a pole-free torus 1-point block, helps to determine numerically 1-point functions of CFTs with rational central charge. For future work, calculations and checks at higher order will help to affirm that the algebraic form holds for orders higher than four. 

Within the second part, we moved to CFTs with affine Lie algebras as underlying symmetry, where we specially focused on null vectors and degenerate representations. In the case of $\hat{\mathfrak{sl}}_2$, we have found a $\mathfrak{sl}_2$-basis independent and horizontal representation independent (up to indecomposability) null operator at level 1 with quadratic Casimir $C_2 = \frac{t^2}{2} - t$,
\begin{equation}
    \hat{T}^c_1 := 2 K_{ab} J^a_{-1}J^b_{0} J_0^c - 
     t f^c_{ab} J^a_{-1}J^b_0 - t^2 J^c_{-1}.
\end{equation}

Applying the null operator $\hat{T}^c_0$ on any state of the horizontal representation, generates a null vector in the usual sense. In this framework we derived the fusion rules of degenerate representations with general highest weight representations, which coincides with earlier work.

\section*{Acknowledgement}

I would like to express my deepest gratitude to my supervisor Dr. Sylvain Ribault, who invested a lot of time to discuss the content of this work and for his all-time valuable advice. A special thanks goes to the Institute de Physique Théoretique for hosting me during the M2-internship and providing me this very special opportunity. Last but not least I want to thank my family without their help I wouldn't be at this point where I am.

\newpage

\appendix

\section{Singularity free at order 3 and 4} \label{appA}

The $c$-pole free expression $K_N$ (\ref{zamoanddaro}) at order 3 and 4 have been calculated and nummerically checked to coincide with the expressions (\ref{horder}) outside the poles in the central charge. To express the solutions, we recall the $b$-symmetric Laurent-polynomials,
\begin{equation}
    B_0 = 1, \quad B_j = \beta^j + \beta^{-j}. 
\end{equation}

At order 3 the solutions reads,
\begin{gather} \label{K3}
    \begin{aligned}
        K_3 &= \frac{1}{192} \bigg\{\bigg(144 \Delta + 288 \Delta^2 - 144 \Delta \Delta_1 + 48 (\Delta_1-1) \Delta_1 + 144 \Delta \Delta_1^2 - 12 (\Delta_1 - 1) \Delta_1^2 \\
        &+ 12 (\Delta_1-1) \Delta_1^3\bigg) B_2 + \bigg(888 \Delta + 960 \Delta^3 - 1840 \Delta \Delta_1 + 440 (\Delta_1-1) \Delta_1 + 1574 \Delta \Delta_1^2 \\
        &- 164 (\Delta_1-1) \Delta_1^2 - 332 \Delta \Delta_1^3 + 68 (\Delta_1-1) \Delta_1^3 + 22 \Delta \Delta_1^4 + 24 \Delta^2 (76 - 75 \Delta_1 + 27 \Delta_1^2)\bigg) B_1 \\
        &+ \bigg(1536 \Delta + 768 \Delta^4 - 3232 \Delta \Delta_1 + 824 (\Delta_1-1) \Delta_1 + 2535 \Delta \Delta_1^2 - 323 (\Delta_1-1) \Delta_1^2 \\
        &- 590 \Delta \Delta_1^3 + 115 (\Delta_1-1) \Delta_1^3 + 39 \Delta \Delta_1^4 + 96 \Delta^3 (24 - 19 \Delta_1 + 3 \Delta_1^2) + 8 \Delta^2 (342 - 411 \Delta_1 \\
        &+ 238 \Delta_1^2 - 20 \Delta_1^3 + \Delta_1^4)\bigg)\bigg\}
    \end{aligned}
\end{gather}

At order 4 the solutions reads,
\begingroup
\allowdisplaybreaks
\begin{subequations} \label{K4}
    
\begin{align*}
        K_4 &= \frac{45}{2048} \bigg(96 \Delta + 288 \Delta^2 + 192 \Delta^3 - 36 \Delta_1 - 168 \Delta \Delta_1 - 144 \Delta^2 \Delta_1 + 52 \Delta_1^2 + 192 \Delta \Delta_1^2 + 144 \Delta^2 \Delta_1^2 \\ &- 33 \Delta_1^3 - 48 \Delta \Delta_1^3 + 19 \Delta_1^4 + 24 \Delta \Delta_1^4 - 3 \Delta_1^5 + \Delta_1^6\bigg) B_4 + \frac{3}{4096}\bigg(45120 \Delta + 141600 \Delta^2 \\ & + 132000 \Delta^3 + 41280 \Delta^4 - 14652 \Delta_1 - 92604 \Delta \Delta_1 - 136152 \Delta^2 \Delta_1 - 61392 \Delta^3 \Delta_1 + 27284 \Delta_1^2 \\ &+ 109180 \Delta \Delta_1^2 + 105984 \Delta^2 \Delta_1^2 + 26832 \Delta^3 \Delta_1^2 - 22281 \Delta_1^3 - 49095 \Delta \Delta_1^3 - 26064 \Delta^2 \Delta_1^3 \\ & + 11363 \Delta_1^4 + 16885 \Delta \Delta_1^4 + 4392 \Delta^2 \Delta_1^4 - 2091 \Delta_1^5 - 1749 \Delta \Delta_1^5 + 377 \Delta_1^6 + 103 \Delta \Delta_1^6\bigg) B_3  \\ & + \frac{1}{8192}\bigg(1423008 \Delta + 4167168 \Delta^2 + 4189536 \Delta^3 + 2013312 \Delta^4 + 446976 \Delta^5 - 442908 \Delta_1 \\ &- 3093480 \Delta \Delta_1 - 5225328 \Delta^2 \Delta_1 - 3529776 \Delta^3\Delta_1 - 993408 \Delta^4 \Delta_1 + 914364 \Delta_1^2 + 3900600 \Delta \Delta_1^2 \\ 
        & + 4702504 \Delta^2 \Delta_1^2 + 2296944 \Delta^3 \Delta_1^2 + 336768 \Delta^4 \Delta_1^2 - 799689 \Delta_1^3 - 2092518 \Delta \Delta_1^3 - 1684440 \Delta^2 \Delta_1^3 \\ 
        &- 450048 \Delta^3 \Delta_1^3 + 391899 \Delta_1^4 + 673314 \Delta \Delta_1^4 + 343312 \Delta^2 \Delta_1^4 + 41280 \Delta^3 \Delta_1^4 - 74427 \Delta_1^5 \\ 
        & - 84018 \Delta \Delta_1^5 - 20280 \Delta^2 \Delta_1^5 + 10761 \Delta_1^6 + 4902 \Delta \Delta_1^6 + 712 \Delta^2 \Delta_1^6\bigg) B_2 + \frac{1}{49152} \bigg(1683456 \Delta^6 \\ 
        &+ 3072 \Delta^5 (3292 - 1669 \Delta_1 + 409 \Delta_1^2) + 768 \Delta^4 (45703 - 32773 \Delta_1 + 17503 \Delta_1^2 - 2532 \Delta_1^3 \\
        &+ 162 \Delta_1^4) + 9 \Delta_1 (-734868 + 1587868 \Delta_1 - 1423299 \Delta_1^2 + 684065 \Delta_1^3 - 130937 \Delta_1^4 + 17171 \Delta_1^5) \\ 
        &+ 64 \Delta^3 (978855 - 1016103 \Delta_1 + 729839 \Delta_1^2 - 190791 \Delta_1^3 + 25799 \Delta_1^4 - 1023 \Delta_1^5 + 29 \Delta_1^6) \\ 
        &+ 8 \Delta^2 (7271424 - 10546749 \Delta_1 + 10186433 \Delta_1^2 - 3976128 \Delta_1^3 +    840752 \Delta_1^4 - 61824 \Delta_1^5 \\ 
        &+ 2156 \Delta_1^6) + 3 \Delta (7167168 - 15881724 \Delta_1 + 20837492 \Delta_1^2 - 11935235 \Delta_1^3 +    3761009 \Delta_1^4 \\ 
        &- 491001 \Delta_1^5 + 28819 \Delta_1^6)\bigg) B_1 + \frac{1}{49152} \bigg(28999296 \Delta + 75605760 \Delta^2 + 85087104 \Delta^3 + 47781120 \Delta^4 \\ 
        &+ 15793152 \Delta^5 + 2617344 \Delta^6 + 344064 \Delta^7 - 8901576 \Delta_1 - 64431288 \Delta \Delta_1 - 115515792 \Delta^2 \Delta_1 \\ 
        &- 92818272 \Delta^3 \Delta_1 - 38587392 \Delta^4 \Delta_1 - 8169984 \Delta^5 \Delta_1 - 1314816 \Delta^6 \Delta_1 + 19491840 \Delta_1^2 \\
        &+ 85511472 \Delta \Delta_1^2 + 114226816 \Delta^2 \Delta_1^2 + 68972000 \Delta^3 \Delta_1^2 + 20325120 \Delta^4 \Delta_1^2 + 3518976 \Delta^5 \Delta_1^2 \\ 
        &+ 208896 \Delta^6 \Delta_1^2 - 17583363 \Delta_1^3 - 49771617 \Delta \Delta_1^3 - 45926220 \Delta^2 \Delta_1^3 - 18616256 \Delta^3 \Delta_1^3 \\ 
        &- 3539200 \Delta^4 \Delta_1^3 - 344064 \Delta^5 \Delta_1^3 + 8397441 \Delta_1^4 + 15592539 \Delta \Delta_1^4 + 9878804 \Delta^2 \Delta_1^4 \\ 
        &+ 2576704 \Delta^3 \Delta_1^4 + 314880 \Delta^4 \Delta_1^4 + 18432 \Delta^5 \Delta_1^4 - 1609353 \Delta_1^5 - 2056035 \Delta \Delta_1^5 - 775652 \Delta^2 \Delta_1^5 \\ 
        &- 102976 \Delta^3 \Delta_1^5 - 9984 \Delta^4 \Delta_1^5 + 205011 \Delta_1^6 + 121329 \Delta \Delta_1^6 + 26924 \Delta^2 \Delta_1^6 + 2880 \Delta^3 \Delta_1^6 + 256 \Delta^4 \Delta_1^6\bigg) \\ \tag{\ref{K4}}
\end{align*}

\end{subequations}
\endgroup

\section{Irreducible representations of $\mathfrak{sl}_2$} \label{appB}

We give a short overview of irreducible representations (shorthand \textit{irreps}) of $\mathfrak{sl}_2$, as examples of horizontal representations of highest weight representations. For simplicity we work in the $m$-basis, which diagonalizes the $J^0_0$ operator. Hence the states are given by $|j,m\rangle$, with diagonal quadratic Casimir operator $C_2 = 2j(j+1)$.
There exist 4 different types of irreps, one finite type with half integer spin $j \in \mathbb{N}/2$ and three infinite types with complex spin $j \in \mathbb{C}$. The three infinite types have either one of the state $|j,\pm j\rangle$ in their space or none, whereas the finite type has both states in its space. In a diagramatic way,

\begin{figure}[H]
    \centering
    \includegraphics[width=12cm]{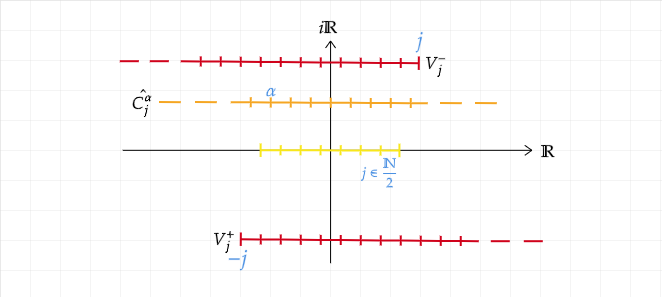}
    \caption{Diagramatic representation of the $J^0_0$-eigenvalue of the four irreducible representations. The yellow coloured line corresponds to half integer spin and is finite (contains $|j,\pm j\rangle$), the red coloured lines correspond to discrete type representations (upper line contains $|j,j \rangle$, lower line $|j,-j\rangle$) and the orange coloured to the continuous type representation (contains none $|j,\pm j\rangle$).}
    \label{fig:sl2reps}
\end{figure}

First we define a general action of the subalgebra $\mathfrak{sl_2} \subset \hat{\mathfrak{sl_2}}$ on the descending vector space,
\begin{equation}
    V^-_j := \mathrm{span}_{\mathbb{C}}(\{|j,m\rangle; j \in \mathbb{C}, m \in j - \mathbb{N}\}).
\end{equation}

The action is defined as,
\begin{gather}
    J^0_0 |j,m\rangle = m |j,m\rangle, \quad J^{\pm}_0 |j,m\rangle = (j\mp m)|j,m\pm 1\rangle.
\end{gather}

This action can be generated by the highest weight state $|j,j\rangle$, which satisfies $ J^+_0|j,j\rangle = 0$.

With the help of the self-inverse automorphism $(\circ)^*: J^0_0 \mapsto - J^0_0, J^{\pm}_0 = - J^{\mp}_0$ we define the ascending vector space implicit via the action $V^+_j := (V^-_j)^*$. Here we use the same notation for the vector space and its representation but mean that we \textit{apply the automorphism action on the descending vector space} and identify $|j,-m\rangle^* = |j,m\rangle$. It is characterised by $m \in -j + \mathbb{N}$ and generated by the lowest weight $|j,-j\rangle^*$ because $0 = (J^-_0)^* |j,j\rangle = J^-_0|j,-j\rangle^*$.

If we restrict on $j \in (-\infty, -1/2)$ we call $D^{\pm}_j := V^{\pm}_j $ the discrete series representations, both are unitary, infinite and irreducible.

The continuous series representation is less constrained and defined via,
\begin{equation}
    C_j^{\alpha} := \mathrm{span}_{\mathbb{C}}(\{|j,m\rangle; j \in -1/2 + i\mathbb{R}_+, \alpha \in \mathbb{R}/\mathbb{Z}, m \in \alpha + \mathbb{Z}\}).
\end{equation}

Where we use the same action as for the discrete series and notice that it has not a highest nor a lowest weight state. It is unitary, infinite and irreducible. Under conjugation we have $(C^{\alpha}_j)^* = C^{-\alpha}_j$.

\end{document}